\documentclass{article}
\usepackage{mathtools, amssymb, mathrsfs, amsthm}
\usepackage[margin=2cm]{geometry}
\usepackage{graphicx}
\graphicspath{{art/}}
\usepackage[font = small]{caption}
\usepackage{enumitem}
\usepackage[utf8]{inputenc}
\usepackage[english]{babel}
\usepackage[colorlinks = true, citecolor = blue, urlcolor = blue]{hyperref}
\usepackage{xcolor}
\usepackage{cleveref}
\usepackage{apptools}
\AtAppendix{\counterwithin{proposition}{section}}

\newtheorem{lemma}{Lemma}
\newtheorem{proposition}{Proposition}
\newtheorem{theorem}{Theorem}
\newtheorem{corollary}{Corollary}
\newtheorem{assumption}{Assumption}
\newtheorem{definition}{Definition}
\theoremstyle{definition}  
\newtheorem{example}{Example}
\newtheorem{remark}{Remark}

\newcommand{\Var}{\textup{Var}}
\newcommand{\E}{\textup{E}}
\newcommand{\F}{\mathcal{F}}

\newcommand{\R}{\mathbb{R}}

\newcommand{\M}{\mathcal{M}}
\def\dm{{ \mathrm{\scriptscriptstyle DM} }}
\def\reg{{ \mathrm{\scriptscriptstyle REG} }}
\def\hc{{ \mathrm{\scriptscriptstyle HC} }}
\newcommand{\T}{\top}

\newcommand{\N}{N}
\renewcommand{\P}{\mathbf{P}}
\renewcommand{\d}{\textup{d}}
\renewcommand{\epsilon}{\varepsilon}
\DeclareMathOperator*{\argmin}{argmin}

\newcommand{\indep}{\rotatebox[origin=c]{90}{$\models$}}
\begin{document}

\title{On the statistical role of inexact matching in observational studies}
\author{Kevin Guo \and Dominik Rothenh\"ausler} 
\date{\today}
\maketitle

\begin{abstract}
In observational causal inference, exact covariate matching plays two statistical roles:  (1) it effectively controls for bias due to measured confounding;  (2) it justifies assumption-free inference based on randomization tests.  This paper shows that \emph{inexact} covariate matching does not always play these same roles.  We find that inexact matching often leaves behind statistically meaningful bias and that this bias renders standard randomization tests asymptotically invalid.  We therefore recommend additional model-based covariate adjustment after inexact matching.  In the framework of local misspecification, we prove that matching makes subsequent parametric analyses less sensitive to model selection or misspecification.  We argue that gaining this robustness is the primary statistical role of inexact matching.
\end{abstract}

\section{Introduction}

\subsection{Motivation}

We consider the problem of using a large observational dataset $\{ (X_i, Y_i, Z_i) \}_{i \leq n}$ to test whether a binary treatment $Z_i \in \{ 0, 1 \}$ has any causal effect on an outcome $Y_i \in \R$.  The vector $X_i \in \R^d$ contains covariates whose confounding effects must be controlled away.  Throughout this article, we assume that all relevant confounders are contained in $X_i$.
 
For over seventy years, statisticians have been tackling such problems using matching methods.  These methods control for the effects of the $X_i$'s by pairing each treated observation $i$ with a similar untreated observation $m(i)$.  Conceptually, this pairing process is often seen as extracting an approximate randomized experiment from an observational dataset \cite{rubin2007, rosenbaum2010, king_nielsen2019}.

When all treated observations are matched exactly ($X_i = X_{m(i)}$), a matched observational study reconstructs a randomized experiment in a statistically precise sense:  conditional on the matches, the treatment distribution among matched units is the same as the treatment distribution in a paired experiment \cite[Section 3.2.3]{rosenbaum2002}.  Using this connection, p-values and confidence intervals can be computed for exactly-matched observational studies using the same design-based randomization tests originally developed for experiments.  These tests only exploit the randomness in $Z$ and are thus valid without any assumptions on the $X$-$Y$ relationship.

However, it is often not possible to find an exact match for every treated observation.  For example, no treated units will be exactly matched when covariates are continuously-distributed.

In the presence of inexact matches, the precise statistical connection between matched-pairs studies and randomized experiments breaks down.  Unlike a paired experiment, the treatment distribution in an inexactly-matched observational study is neither uniform within pairs \cite{hansen2009} nor independent across pairs \cite[Section 5]{pashley_etal2021}.  As a result, standard randomization tests based on uniformity and independence lose their finite-sample validity.  Indeed, \cite{shah_peters2020} show that no nontrivial test can have assumption-free validity when continuous covariates are present.

This paper asks what statistical role matching plays when not all units can be matched exactly.  We consider two main possibilities.
\begin{enumerate}[label=(\Alph*)]
    \item Perhaps matching discrepancies typically become negligible in large samples, so that standard randomization tests remain asymptotically valid despite not being finite-sample exact.  If so, then the statistical role of inexact matching would be the same as the statistical role of exact matching:  controlling overt bias and providing a basis for nonparametric inference.  \label{item:frt_justification}
    \item Alternatively, matching discrepancies could remain statistically meaningful even in large samples and render standard randomization tests invalid.  If so, then additional model-based adjustment after matching would be necessary to obtain valid inference.  In this case, \cite{ho_imai_king_stuart_2007} has argued that the statistical role of matching is to provide a pre-processing step that makes subsequent model-based inferences less sensitive to model selection or misspecification. \label{item:robustness_justification}
\end{enumerate}

\subsection{Outline and overview of results}

In the first half of this article, we investigate possibility \ref{item:frt_justification} by studying the large-sample properties of randomization tests in matched observational studies.  Our formal results are developed for \cite{rosenbaum1989optimal}'s optimal Mahalanobis matching scheme.  However, much of the intuition extends to other matching schemes.

We find that conventional randomization tests are not generally valid in large samples, even under strong smoothness assumptions.  In fact, their Type I error may be dramatically inflated even when the true outcome model is linear, only a handful of covariates are present, and conventional balance tests pass.  The main issue is that covariate matching does not eliminate bias at a fast rate.  A secondary issue is that randomization tests may underestimate the sampling variance of commonly-used test statistics.  Previously, \cite{abadie_imbens2006, savje2021} and others have reported on this bias, although the variance issue seems to be a new finding.

Thus, we caution against applying standard randomization tests after matching inexactly.  Although the idea that matching approximates a randomized experiment is a useful conceptual tool \cite{rubin2007, rubin2008}, the analogy is often not precise enough to form the basis for inference.

In the second half of this article, we argue that \ref{item:robustness_justification} provides a more compelling justification for inexact matching.  We prove that in an appropriately-matched dataset, p-values based on linear regression remain approximately valid even if the linear model is locally misspecified.  Moreover, after matching, all sufficiently accurate model specifications will yield nearly identical $t$-statistics.  These results give formal support to claims made in \cite{rubin1973regression, rubin1979regression, ho_imai_king_stuart_2007}, and others.  However, our analysis gives additional insights.  In particular, we find that it is generally necessary to use matching with replacement rather than pair matching to achieve the full extent of robustness attainable by matching.

Based on these results, we recommend model-based adjustment and inference after matching inexactly.  Conceptually, this mode of inference makes transparent that structural assumptions such as approximate linearity are still required for reliable inference after inexact matching.  It also cleanly separates the randomness used for study design ($X_i$ and $Z_i$) from the randomness used in outcome analysis ($Y_i$).  The Bayesian approach advocated by \cite{rubin1991modes} also has these conceptual advantages, but this article focuses on frequentist inference.

In summary, our findings suggest re-thinking the role of inexact matching in observational studies.  On its own, covariate matching may not remove enough bias to justify the use of assumption-free randomization tests.  Moreover, inexact matching may lead randomization tests to underestimate the sampling variability of common test statistics.  However, matching does play an important role in the design stage, by making downstream parametric analyses more robust to model selection or misspecification.

\subsection{Setting}

The setting of this article is the Neyman-Rubin causal model with an infinite superpopulation.  We assume that units $\{ (X_i, Y_i(0), Y_i(1), Z_i) \}$ are independent samples from a common distribution $P$ and that only $(X_i, Y_i, Z_i)$ is observed, where $Y_i = Y_i(Z_i)$.  The problem of interest is to use the observed data to test Fisher's sharp null hypothesis:
\begin{align}
    H_0 \, : \, Y_i(0) = Y_i(1) \text{ with probability one under $P$.} \label{fishers_null}
\end{align}
All our results extend to testing a constant treatment effect, $H_{\tau} \, : \, Y_i(0) + \tau = Y_i(1)$.  However, they will not extend to the weak null hypothesis $\E \{ Y_i(1) - Y_i(0) \} = 0$.

Throughout, we assume the underlying population satisfies a few conditions:

\begin{assumption}
\label{assumption:primitives}
The distribution $P$ satisfies the following:
\begin{enumerate}[label=(\alph*), itemsep=-0.75ex]
    \item \emph{Unconfoundedness}. $\{ Y(0), Y(1) \} \, \indep \, Z \mid X$.
    \item \emph{Overlap}. $P(Z = 0 \mid X) \geq \delta > 0$. \label{item:overlap}
    \item \emph{More controls than treated}. $0 <  P(Z = 1) < 0.5$.
    \item \emph{Moments}. $\| X \|$ and $Y$ have more than four moments. \label{item:four_moments}
    \item \emph{Nonsingularity}. $\Var(X \mid Z = 1) \succ \mathbf{0}$ and $\Var(Y \mid X, Z) > 0$. 
\end{enumerate}
\end{assumption}

Unconfoundedness and overlap are standard identifying assumptions.  Meanwhile, the condition $0 < P(Z = 1) < 0.5$ ensures that treated observations exist and that it is eventually possible to find an untreated match for each treated observation.  The last two conditions are needed for various technical reasons, e.g. to ensure that the Mahalanobis distance exists.  

The analysis in this paper is asymptotic, and we assume that the sample size $n$ grows large as the dimension $d$ stays fixed.  In fact, following \cite{rubin1980bias}'s advice, we recommend thinking of $d$ as a fairly small number, say, eight or less.

The key asymptotic concept studied in this paper is the \emph{asymptotic validity} of p-values.

\begin{definition}[Asymptotic validity]
\label{definition:asymptotic_validity}
A sequence of p-values $\hat{p}_n$ is called \emph{asymptotically valid} at $P \in H_0$ if (\ref{eq:asymptotic_validity}) holds under independent sampling from $P$.
\begin{align}
\limsup_{n \rightarrow \infty} \P( \hat{p}_n < \alpha) \leq \alpha \quad \text{for every} \quad \alpha \in (0, 1) \label{eq:asymptotic_validity}
\end{align}
\end{definition}

\begin{remark}[Alternative sampling models]
The independent sampling model used in this paper differs from several alternatives used in the matching literature.  One alternative is the design-only framework which models $Z_i$ as random but treats both the matching and the unit characteristics $\{ X_i, Y_i(0), Y_i(1) \}$ as fixed \cite{rosenbaum2002}.  While this simplifies many issues, it precludes analyzing the typical size of matching discrepancies.  It also assumes away the complex dependence between the treatments $Z_i$ and the matching $\mathcal{M}$, which may be practically relevant \cite{pimentel2022}.  Another alternative assumes that the number of untreated observations $N_0$ grows much faster than the number of treated observations $N_1$.  For example, \cite{ferman2021matching} and \cite{abadie_imbens2012} assume that $N_0 \gg N_1^{d/2}$.  This scaling is favorable for matching, but the sample size requirement is stringent even for $N_1 = 100, d = 5$.  We find that the standard sampling regime gives better approximations in problems where $N_0$ is only a constant multiple of $N_1$.
\end{remark}

\section{Large-sample properties of paired randomization tests} \label{section:paired}

\subsection{Optimal matching and Fisher's randomization test}

In the first part of this paper, we present our findings on the large-sample properties of standard randomization tests in inexactly-matched observational studies.

The pair matching procedure we study is \cite{rosenbaum1989optimal}'s optimal Mahalanobis matching scheme.  This matching scheme pairs each treated observation $i$ with a unique untreated observation $m(i)$ in a way that minimizes the total Mahalanobis distance across pairs:
\begin{align}
\sum_{Z_i = 1} \{ (X_i - X_{m(i)})^{\T} \mathbf{\hat{\Sigma}}^{-1} (X_i - X_{m(i)}) \}^{1/2} \label{total_mahalanobis}
\end{align}
Ties may be broken arbitrarily.  In \Cref{total_mahalanobis}, $\mathbf{\hat{\Sigma}}$ denotes the sample covariance matrix of $X$ and we arbitrarily set $\mathbf{\hat{\Sigma}}^{-1} = \mathbf{I}_{d \times d}$ when $\hat{\mathbf{\Sigma}}$ is singular.  We also let $\mathcal{M} = \{ i \, : \, Z_i = 1 \text{ or } i = m(j) \text{ for some treated unit $j$} \}$ denote the set of matched units.

The randomization test we study is the paired Fisher randomization test.  This test computes a p-value for the null hypothesis (\ref{fishers_null}) as follows.  First, the user computes a test statistic $\hat{\tau} \equiv \hat{\tau}( \{ (X_i, Y_i, Z_i) \}_{i \in \mathcal{M}})$ on the matched data.  Two widely-used test statistics are the difference-of-means statistic (\ref{eq:dm_statistic}) and the regression-adjusted statistic (\ref{eq:ols_statistic}).
\begin{align}
    \hat{\tau}^{\dm} &= \frac{1}{N_1} \sum_{Z_i = 1} \{ Y_i - Y_{m(i)} \} \label{eq:dm_statistic}\\
    \hat{\tau}^{\reg} &= \argmin_{\tau \in \R} \min_{(\gamma, \beta) \in \R^{1 + d}} \sum_{i \in \mathcal{M}} (Y_i - \gamma - \tau Z_i - \beta^{\top} X_i)^2 \label{eq:ols_statistic}
\end{align}
Then, conditional on the original data $\mathcal{D}_n = \{ (X_i, Y_i, Z_i) \}_{i \leq n}$, the user defines pseudo-assignments $\{ Z_i^* \}_{i \in \mathcal{M}}$ by randomly permuting the true assignments $Z_i$ across matched pairs.  Finally, the p-value is defined as $\hat{p} = \P( | \hat{\tau}_*| \geq | \hat{\tau}| \mid \mathcal{D}_n)$ where $\hat{\tau}_* \equiv \hat{\tau}( \{ (X_i, Y_i, Z_i^*) \}_{i \in \mathcal{M}})$ is the test statistic evaluated using the pseudo-assignments instead of the true ones.  When there are more treated than control units or no treated units, we arbitrarily set $\hat{p} = 1$ since $\mathcal{M}$ is undefined.

The distribution of $\hat{\tau}_*$ given $\mathcal{D}$ is called the \emph{randomization distribution} of $\hat{\tau}_*$.  In practice, this distribution will be approximated using randomly sampled permutations.  Since the approximation error can be made arbitrarily small by sampling a large number of permutations, we will consider the idealized case where randomization probabilities are computed exactly.

\subsection{The paired Fisher randomization test is not generally valid} \label{section:examples_of_invalidity}

In this section, we give theoretical and numerical examples showing that the paired Fisher randomization test may fail to control asymptotic Type I error even in problems with smooth propensity scores and outcome models.  We also explain what goes wrong in each example.  All formal claims are proved in the Supplementary Materials.

Throughout, we use the following notations:  $e(x) := P(Z = 1 \mid X = x)$ is the propensity score, $\hat{p}^{\dm}$ is the randomization p-value based on the difference-of-means statistic (\ref{eq:dm_statistic}), and $\hat{p}^{\reg}$ is the randomization p-value based on the regression-adjusted statistic (\ref{eq:ols_statistic}).

\begin{example}[One covariate] \label{example:one_covariate}
The first example is based on the analysis in \cite{savje2021}.  Suppose that $P \in H_0$ satisfies Assumption \ref{assumption:primitives} and the following:
\begin{align*}
X &\sim \textup{Uniform}(0, 1)\\
Z \mid X &\sim \textup{Bernoulli}(\theta_0 + \theta_1 X)\\
Y \mid X, Z &\sim N( \beta_0 + \beta_1 X, \sigma^2).
\end{align*}
If $\beta_1 \neq 0$ and $P \{ e(X) \geq 0.5 \} > 0$, then $\P ( \hat{p}^{\dm} < \alpha) \rightarrow 1$ for every $\alpha \in (0, 1)$.  In other words, if overt bias is present and any units in the population have propensity scores larger than one-half, then the paired Fisher randomization test will almost always make a false discovery.  In this example, the same conclusion would hold for optimal propensity score matching.
\end{example}

\begin{example}[Two covariates] \label{example:two_covariates}
A multivariate analogue of \Cref{example:one_covariate} can be constructed using the method from \cite{rubin1976}.  Let $\mathbf{D} = \{ x \in \R^2 \, : \, \| x \|_2 \leq 1 \}$ be the unit disc in the plane, and let $P \in H_0$ be any distribution satisfying Assumption \ref{assumption:primitives} and the following:
\begin{align*}
X &\sim \textup{Uniform}(\mathbf{D})\\
Z \mid X &\sim \textup{Bernoulli} \{ 0.35(1 + \theta^{\top} X) \}\\
Y \mid X, Z &\sim N( \theta^{\top} X, \sigma^2)
\end{align*}
for some $\theta \in \mathbf{D}$.  A typical large sample from this distribution will have nearly twice as many untreated units as treated units.  However, some of those units will have propensity scores larger than one-half.  As a result, we still have $\P( \hat{p}^{\dm} < \alpha) \rightarrow 1$ for every $\alpha \in (0, 1)$.  The same conclusion holds under optimal propensity score matching, nearest-neighbor matching, or any other maximal pair-matching scheme.
\end{example}

In both of these examples, the paired Fisher randomization test fails because the test statistic $\hat{\tau}^{\dm}$ is asymptotically biased, but the randomization distribution does not replicate this bias.  This bias is mainly driven by the presence of units with propensity scores larger than one-half, because any pair-matching scheme must eventually run out of close matches for these units.  After all, treated units outnumber untreated units in regions of covariate space with $e(x) > 0.5$.  See \cite{savje2021} for a careful discussion of this issue.

In large samples, a careful analyst may detect this bias and attempt to remove it using regression adjustment.  For example, \cite{carpenter1977} writes: ``if the residual bias after matching is unacceptably large it may be removed by analysis of covariance."  However, the following example shows that this may not be enough to rescue the randomization p-value.

\begin{example}[Regression adjusted test statistics] \label{example:regression_adjusted_test_statistics}
Let $P$ satisfy the requirements of \Cref{example:one_covariate}, including $P \{ e(X) \geq 0.5 \} > 0$.  Consider the paired Fisher randomization test based on the regression-adjusted test statistic $\hat{\tau}^{\reg}$.  Standard least-squares theory tells us that this test statistic is exactly unbiased in finite samples.  Nevertheless, we have:
\begin{align*}
\limsup_{n \rightarrow \infty} \P( \hat{p}^{\reg} < \alpha) > \alpha \quad \text{for every } \alpha \in (0, 1).
\end{align*}
Thus, even the paired Fisher randomization test based on a correctly-specified regression model does not control asymptotic Type I error when units with propensity scores larger than one-half are present.
\end{example}

The issue here is more subtle, and is caused by a variance mismatch.  Since matching fails to balance covariates when $P \{ e(X) \geq 0.5 \} > 0$, the sample correlation between $X_i$ and $Z_i$ in $\mathcal{M}$ does not vanish in large samples.  However, $X_i$ and $Z_i^*$ are uncorrelated in the randomization distribution.  Correlation harms precision in least-squares regression, so this mismatch leads the randomization variance of $\hat{\tau}_*^{\reg}$ to underestimate the sampling variance of $\hat{\tau}^{\reg}$.

Rather than using regression-adjusted test statistics, some authors have recommended only using pair matching in populations where all units have propensity scores less than one-half.  For example, \cite{yu_silber_rosenbaum2020rejoinder} write ``In concept in large samples, pair matching is feasible if $\{ 1 - e(x) \} / e(x) > 1$ for all $x$" (notation edited to match ours).  

In such populations, it is eventually possible to find an arbitrarily close match for every treated unit.  As a result, $\hat{\tau}^{\dm}$ will be asymptotically unbiased and consistent.  However, asymptotic unbiasedness is not enough to justify randomization tests.  Valid inference require biases to be so small that ``that they are buried in estimated standard errors" \cite{rubin2022interview}.  Since the standard errors of the randomization distribution tend to zero at rate $n^{-1/2}$ \cite{bai2021matched}, valid inference requires bias to decay at a rate faster than $n^{-1/2}$.  This is stated formally in the following Proposition.

\begin{proposition}[Bias requirement]
\label{proposition:balance_requirement}
Suppose that $P \in H_0$ satisfies Assumption \ref{assumption:primitives} and $P \{ e(X) < 0.5 \} = 1$.  Then the randomization p-value $\hat{p}^{\dm}$ is asymptotically valid if and only if $\E( \hat{\tau}^{\dm} \mid \{ (X_i, Z_i) \}_{i \leq n}) = o_P(n^{-1/2})$.
\end{proposition}

When the linear model $\E(Y \mid X, Z) = \gamma + \beta^{\top} X$ holds, \Cref{proposition:balance_requirement} requires that optimal matching achieves very fine covariate balance in the direction of $\beta$:
\begin{align}
\frac{1}{N_1} \sum_{Z_i = 1} \{ X_i - X_{m(i)} \}^{\top} \beta = o_P(n^{-1/2}). \label{balance_requirement}
\end{align}
Unless the dimension $d$ is very small, \cite[Theorem 2.(ii)]{abadie_imbens2006} suggests that such strict balance is hard to achieve.  This is illustrated by the following numerical example.

\begin{example}[Slow bias decay]
\label{example:dimension4}
For various sample sizes $n$ between $200$ and $2,000$, we sampled data from the following linear/logistic model:
\begin{align*}
X &\sim \textup{Uniform}([-1, 1]^4)\\
Z \mid X & \sim \textup{Bernoulli}[1/\{ 1 + \exp(1.1 - X_1) \}]\\
Y \mid X, Z &\sim \N(3 X_1, 1).
\end{align*}
This was repeated $2,000$ times per sample size.  In each simulation, we recorded the bias of $\hat{\tau}^{\dm}$ and the paired randomization test p-value.  The results are shown in \Cref{fig:dimension_example}.  Although this distribution has $P \{ e(X) < 0.475 \} = 1$, the paired randomization test nevertheless performs poorly due to the failure of the bias condition (\ref{balance_requirement}).  In fact, the Type I error of nominally level 5\% tests appears to increase with the sample size.
\end{example}

\begin{figure}
\centering
    \includegraphics[width=14cm]{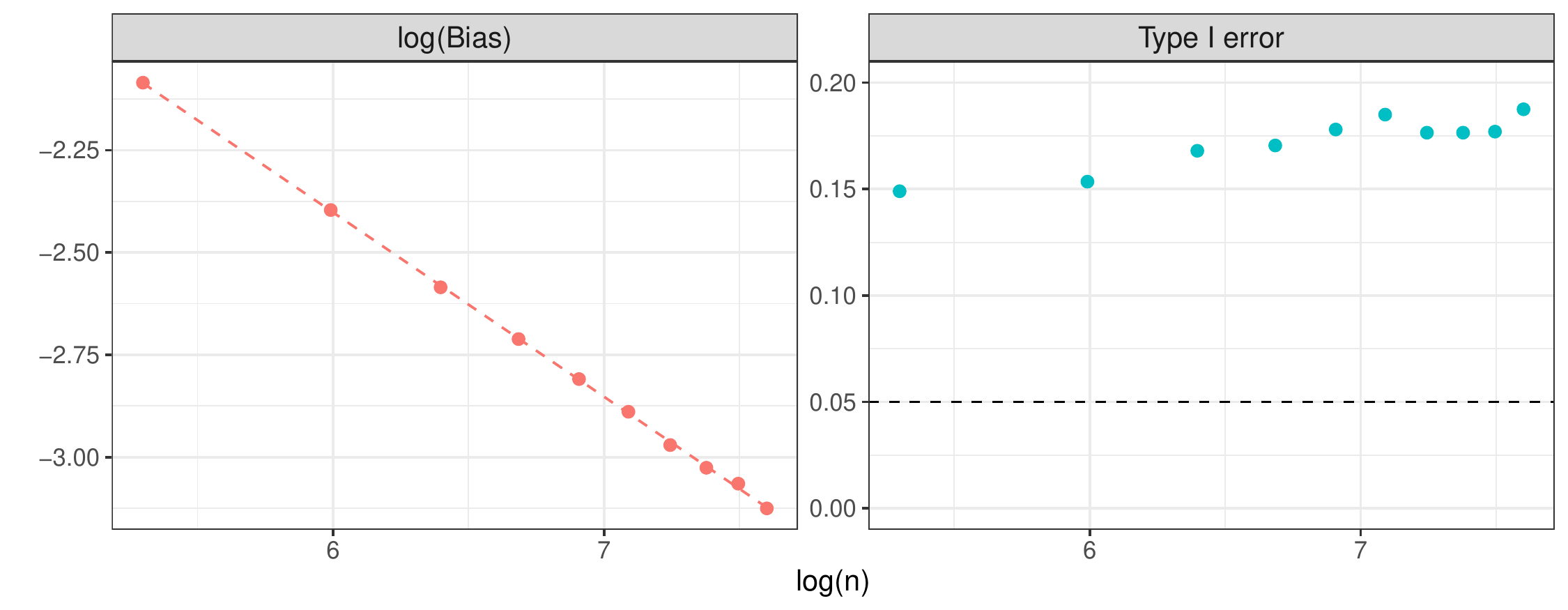}
    \caption{The left panel plots the average bias of the difference-of-means statistic in \Cref{example:dimension4} at various sample sizes, on a log-log scale.  The slope of the best-fit line is $\approx -0.45$, suggesting that the bias does not satisfy the $o(n^{-1/2})$ decay rate required by \Cref{proposition:balance_requirement}.  The right panel shows the Type I error of nominally level 5\% paired randomization tests based on the difference-of-means statistic.  All matches were computed in the $\texttt{R}$ programming language \cite{rcoreteam} using the \texttt{optmatch} package \cite{optmatch}.  Randomization p-values were approximated using 1,000 randomly sampled permutations.}
    \label{fig:dimension_example}
\end{figure}

\begin{remark}[Balance tests]
It is worth mentioning that the balance condition (\ref{balance_requirement}) would not hold even in a completely randomized experiment, so the asymptotic validity of $\hat{p}^{\dm}$ cannot be certified by any balance test with a completely randomized reference distribution.  This includes the two-sample $t$-test and all of the examples in \cite[Chapter 10]{rosenbaum2010}.  Indeed, in each of our simulations in \Cref{example:dimension4}, we also performed a nominal level 10\% balance check using Hotelling's $T^2$ test.  Imbalance was not detected even a single time.  A balance test based on a \emph{paired} experiment reference distribution would be powerful enough to detect cases where $\P( \hat{p}^{\dm} < \alpha) \rightarrow 1$ \cite{hansen2009}, although not powerful enough certify $\P( \hat{p}^{\dm} < \alpha) \rightarrow \alpha$.   See \cite{austin2008, hansen2008, imai_king_stuart2007, branson2021randomization} for further discussion of balance tests.
\end{remark}

\begin{remark}[Calipers]
Some of the poor behavior in these examples might be avoided or mitigated using calipers on the propensity score or the raw covariates \cite{hansen2009}.  However, the correct scaling of this caliper is a delicate issue which goes beyond the scope of this paper.  In \cite{pimentel2022}'s simulations, propensity calipers did help to reduce the false positive rate of the paired Fisher randomization test.  On the other hand, \cite{schafer_kang2009}'s simulations used propensity calipers but still found $\hat{\tau}^{\dm}$ to be severely biased.   \cite{kim2021local} study randomization tests based on finely-calipered stratifications of the covariate space.  Their results suggest that in moderate dimensions, obtaining validity via covariate calipers may require discarding the vast majority of treated observations.
\end{remark}

\subsection{Sufficient conditions for validity} \label{section:sufficient_conditions}

The examples in the previous section show that stringent sampling assumptions are required for the paired Fisher randomization test to be asymptotically valid.  For completeness, this section gives two sets of sufficient conditions that make this work.

First, we consider the test based on the difference-of-means statistic.  To control the bias that spoils validity in Examples \ref{example:one_covariate} and \ref{example:two_covariates}, we must assume that no units have propensity scores above one-half.  However, \Cref{example:dimension4} shows that this is not enough and we also need to restrict attention to very low-dimensional problems.

\begin{proposition}[Sufficient conditions for $\hat{p}^{\dm}$] \label{proposition:dm_sufficient}
Let $P \in H_0$ satisfy Assumption \ref{assumption:primitives} and the following:
\begin{enumerate}[label=(\alph*), itemsep=-0.75ex]
    \item \emph{No large propensity scores}. $P\{ e(X) < 0.5 - \eta \} = 1$ for some $\eta > 0$. \label{item:no_large_propensities}
    \item \emph{Smooth outcome model}. The map $x \mapsto \E(Y \mid X = x)$ is Lipschitz-continuous. \label{item:lipschitz_outcome_model}
    \item \emph{One continuous covariate}.  $X_i$ is scalar-valued, and has a continuous, positive density supported on a compact interval. \label{item:one_continuous_covariate}
\end{enumerate}
Then the p-value $\hat{p}^{\dm}$ based on the difference-of-means statistic (\ref{eq:dm_statistic}) is asymptotically valid.
\end{proposition}

We conjecture that asymptotic validity continues to hold with up to three continuous covariates.  However, proving this is beyond our current abilities.  In dimension four, \Cref{example:dimension4} suggests that asymptotic validity will no longer hold.

Next, we consider the test based on the regression-adjusted test statistic.  If we assume the model is correctly specified, then the bias of $\hat{\tau}^{\reg}$ is controlled even if more than three continuous covariates are present.  Meanwhile, the variance mismatch in \Cref{example:regression_adjusted_test_statistics} can be ruled out by assuming that no units have propensity scores larger than one-half.

\begin{proposition}[Sufficient conditions for $\hat{p}^{\reg}$] \label{proposition:ols_sufficient}
Let $P \in H_0$ satisfy Assumption \ref{assumption:primitives} and the following:
\begin{enumerate}[label=(\alph*), itemsep=-0.75ex]
    \item \emph{No large propensity scores}.  $P \{ e(X) < 0.5 - \eta \} = 1$ for some $\eta > 0$. \label{item:no_large_propensities_ols}
    \item \emph{Correctly-specified outcome model}. $\E(Y \mid X, Z) = \gamma + \beta^{\top} X$ for some $(\gamma, \beta) \in \R^{1 + d}$. \label{item:correctly_specified_parametric_model}
\end{enumerate}
Then the p-value $\hat{p}^{\reg}$ based on the regression-adjusted test statistic (\ref{eq:ols_statistic}) is asymptotically valid.
\end{proposition}

By appropriately modifying our proofs, the same conclusion can be extended to other correctly-specified parametric regression models, e.g. logistic regression.  However, under assumption \ref{item:correctly_specified_parametric_model}, it is not necessary to use randomization inference for hypothesis testing.  Model-based ``sandwich" standard errors would work just as well if not better, since they remain valid even when the propensity score condition \ref{item:no_large_propensities_ols} fails.  Meanwhile, \Cref{example:regression_adjusted_test_statistics} shows that Fisher's randomization test may be invalid when large propensity scores are present.

\begin{remark}[Randomization tests vs. randomization inference] \label{remark:tests_vs_inference}
Although the results in this section provide some justification for randomization tests, the justifications are not truly design-based.  The key principle of design-based randomization inference is to base probability statements on the conditional randomness in $Z_i$ given everything else.  However, as \cite{pimentel2022} and \cite{pashley_etal2021} point out, the conditional distribution of $(Z_i)_{i \in \mathcal{M}}$ given $\mathcal{M}$ and $\{ (X_i, Y_i(0), Y_i(1) \}_{i \leq n}$ is highly intractable unless all matches are exact.  To get around this, the proofs of Propositions \ref{proposition:dm_sufficient} and \ref{proposition:ols_sufficient} actually \emph{condition} on treatments and use the outcome as the source of randomness.  In other words, the justification has nothing to do with design.
\end{remark}

\section{An alternative role for matching} \label{section:alternative_role}

\subsection{Combine matching with parametric outcome modeling}

In the second half of this article, we recommend an alternative framework for inference after matching.  Specifically, we suggest viewing matching as a pre-processing step for a conventional statistical analysis based on parametric outcome models.  This type of post-matching analysis has been recommended by \cite{ho_imai_king_stuart_2007} and \cite{stuart2010}.

The leading example we have in mind is an analysis that fits the linear regression model (\ref{eq:linear_regression}) in the matched sample, interrogates the linearity assumption using specification tests or diagnostic plots, and reports inferential summaries for the coefficient $\hat{\tau}^{\reg}$ based on heteroskedasticity-consistent robust standard errors.
\begin{align}
(\hat{\tau}^{\reg}, \hat{\gamma}, \hat{\beta}) &= \argmin_{(\tau, \gamma, \beta)} \sum_{i \in \mathcal{M}} (Y_i - \gamma - \tau Z_i - \beta^{\top} X_i)^2 \label{eq:linear_regression}
\end{align}

Let us give some motivation for this approach.  It is well-known that accurate outcome modeling improves a matched analysis by cleaning up the residual imbalances that remain after matching \cite{rubin1973regression, rubin1979regression}.  It may be less clear what role matching plays in improving an outcome analysis that already makes parametric assumptions.  For example, \cite{hill2002} asks: ``If the response surfaces are linear why wouldn't standard regression work just as well for covariance adjustment, even perhaps more efficiently than [matching methods]?"

Our main contribution in this half of the article is to prove that matching improves parametric outcome analysis by reducing sensitivity to model selection and misspecification.  Indeed, we regard this as the primary statistical role of matching.  Prior empirical work has made the same point, using a combination of simulations and informal arguments.  However, our formal analysis leads to additional insights.  For example, we show that more robustness is gained from matching with replacement than from optimal pair matching.

\subsection{The local misspecification framework}

To study the role of misspecification, we consider a class of nonlinear models defined through small perturbations of a baseline linear model. 

Let $P \in H_0$ be some distribution satisfying the linear outcome model $\E(Y \mid X, Z) = \gamma + \beta^{\top} X$.  For any bounded nonlinear function $g : \R^d \rightarrow \R^k$, let $P_{h,g}$ be the distribution of the vector $\{ X, Y(0) + h^{\top} g(X), Y(1) + h^{\top} g(X), Z \}$ when $\{ X, Y(0), Y(1), Z \} \sim P$.  Since we have simply added the same nonlinearity to both potential outcomes, Fisher's sharp null hypothesis (\ref{fishers_null}) continues to hold under $P_{h,g}$.  However, the outcome model now contains a nonlinear term:
\begin{align}
\E_{h,g}(Y \mid X, Z) = \gamma + \beta^{\top} X + h^{\top} g(X).
\end{align}

A sequence of models $\{ P_{h_n,g} \}_{n \geq 1}$ is called \emph{locally misspecified} if the coefficient $h \equiv h_n$ tends to zero with the sample size at rate $n^{-1/2}$.  This scaling is meant to model problems where nonlinearities are large enough to affect inference, but not so large that they can easily be caught.  In smooth models, no specification test can consistently detect the nonlinearity in a locally misspecified sequence \cite[Lemma A.1]{leeb_potscher2006}.  

We say that a sequence of p-values is \emph{robust to local misspecification} near $P \in H_0$ if it remains asymptotically valid even when the linear model is locally misspecified.  A more formal definition is the following.

\begin{definition}[Locally robust p-values] \label{definition:locally_robust_pvalue}
A sequence of p-values $\hat{p}_n$ is called \emph{robust to local misspecification} near $P \in H_0$ if (\ref{eq:local_size}) holds for every radius $C < \infty$ and every bounded nonlinear function $g$.
\begin{align}
\limsup_{n \rightarrow \infty} \sup_{\| h \| \leq C n^{-1/2}} P_{h_n,g}^n ( \hat{p}_n < \alpha) \leq \alpha \quad \text{for all } \alpha \in (0, 1). \label{eq:local_size}
\end{align}
\end{definition}

Outside of exceptional cases, p-values based on parametric outcome models fit to the full unmatched sample are not robust to local misspecification.  In contrast, p-values based on best-performing semiparametric methods \cite{vanderLaan_Rose2011, robins_etal2008} typically achieve guarantees far stronger than (\ref{eq:local_size}).  We will see that parametric tests gain some of the robustness of semiparametric methods when the data is first pre-processed using matching.

\subsection{Matching protects against local model misspecification}

The first main result of this section shows that when $P \{ e(X) < 0.5 \} = 1$, model-based p-values computed after optimal Mahalanobis matching remain valid even if the model is locally misspecified.

\begin{theorem}[Matching confers local robustness] \label{theorem:paired_local_robustness}
Let $P \in H_0$ satisfy Assumption \ref{assumption:primitives}, $P \{ e(X) < 0.5 \} = 1$, and the linear outcome model $\E(Y \mid X, Z) = \gamma + \beta^{\top} X$.  Let $\hat{p}^{\hc}$ be the (one- or two-sided) robust standard error p-value for testing the coefficient $\hat{\tau}^{\reg}$ in the regression (\ref{eq:linear_regression}).  Then $\hat{p}^{\hc}$ is robust to local misspecification near $P$ in the sense of \Cref{definition:locally_robust_pvalue}.
\end{theorem}

The intuitive explanation for this robustness is that regression after matching combines two complementary methods of bias reduction.  The first is the nearly-correctly-specified outcome model, which eliminates most of the bias and gets us within an $O(n^{-1/2})$ neighborhood of the correct answer.  From there, the nonparametric consistency of matching kicks in to handle the residual nonlinearity.  This is conceptually similar to the doubly-robust estimator of \cite{robins_rotnitzky_zhao1994}, which combines an outcome model and a propensity model to gain robustness and efficiency.  However, regression after optimal Mahalanobis matching does not produce a consistent estimate of the propensity score, so inferences based on \Cref{theorem:paired_local_robustness} are not semiparametrically efficient.  See \cite{lin2021densityratio} for related discussion.

Unfortunately, the conclusion of \Cref{theorem:paired_local_robustness} does not extend to populations where some units have propensity scores above one-half.  Pair matching may still improve robustness in such problems, but it will not protect against all directions of local misspecification.  The reason is that pair matching runs out of controls in some parts of the covariate space, costing one of the bias reduction methods used in \Cref{theorem:paired_local_robustness}.

Our next result shows that this problem can be avoided by matching \emph{with replacement}:
\begin{align}
m_r(i) \in \argmin_{j \, : \, Z_j = 0} (X_i - X_j)^{\top} \hat{\mathbf{\Sigma}}^{-1} (X_i - X_j). \label{eq:matching_with_replacement}
\end{align} 
To account for the fact that the same control unit may be matched more than once, we also replace the ordinary least-squares regression (\ref{eq:linear_regression}) by a weighted least-squares regression with multiplicity-counting weights. 

\begin{theorem}[Replacing controls helps] \label{theorem:replaced_robustness}
Let $P \in H_0$ satisfy Assumption \ref{assumption:primitives} and the linear outcome model.  Let $\mathcal{M}_r$ be the set of units matched by the scheme (\ref{eq:matching_with_replacement}).  Let $W_i = 1$ if observation $i$ is treated and otherwise set $W_i = \sum_{j = 1}^n Z_j \mathbf{1} \{ m_r(j) = i \}$.  Let $\hat{p}^{\hc}$ be the one- or two-sided robust standard error p-value for testing the coefficient $\hat{\tau}^{\reg}$ in the weighted regression (\ref{eq:weighted_regression}).
\begin{align}
( \hat{\tau}^{\reg}, \hat{\gamma}, \hat{\beta}) = \argmin_{(\tau, \gamma, \beta)} \sum_{i \in \mathcal{M}_r} W_i (Y_i - \gamma - \tau Z_i - \beta^{\top} X_i)^2. \label{eq:weighted_regression}
\end{align}
Then $\hat{p}^{\hc}$ is robust to local misspecification in the sense of \Cref{definition:locally_robust_pvalue}.
\end{theorem}

The reason matching with replacement helps is that it ensures no region of the covariate space will run out of untreated units.  Therefore, the bias-correction opportunity from matching is present even when $P \{ e(X) < 0.5 \} \neq 1$.  Based on this result, we generally recommend matching with replacement over pair matching unless there is good reason to believe that no units have propensity scores larger than one-half.  

The combination of matching with replacement and weighted linear regression has previously been implemented by \cite{dehejia_wahba2002}.  In the \texttt{R} programming language, matching with replacement is implemented by default in \cite{sekhon2011}'s \texttt{Matching} package.  

\subsection{Matching reduces model dependence}

Finally, we show that matching makes parametric analyses less sensitive to the exact model specification.  This gives rigorous support to the main claim in \cite{ho_imai_king_stuart_2007}.

Let $P_{h_n,g}$ be a locally misspecified sequence centered around a baseline linear model $P \in H_0$.  Thus, when the sample size is $n$, the true regression model takes the form:
\begin{align}
\E_{h_n,g}(Y \mid X, Z) = \gamma + \beta^{\top} X + h_n^{\top} g(X). \label{eq:local_outcome_model}
\end{align}
for some sequence $h_n = O(n^{-1/2})$.  We further assume that $\Var[ \{ X, g(X) \} \mid Z = 1] \succ \mathbf{0}$, so that $g$ is genuinely nonlinear in places with treated observations.

Consider three different modeling strategies that might be used to analyze the matched data:
\begin{enumerate}[itemsep=-0.75ex]
    \item \emph{Baseline}.  The first procedure fits a regression model that controls for $X$ linearly.  In $\texttt{R}$ and $\texttt{S}$ formula notation, this procedure fits the model $Y \sim 1 + Z + X$.
    \item \emph{Saturated}.  The second procedure fits a model that correctly includes the nonlinearity in \Cref{eq:local_outcome_model}, $Y \sim 1 + Z + X + g(X)$. 
    \item \emph{Model selector}.  The final procedure fits the saturated model, drops insignificant components of $g$, and then makes inferences as if the chosen model were pre-specified.  We make no assumptions on what significance tests are used in the model pruning step.
\end{enumerate}
After model specification, each procedure produces a p-value based on \cite{white1980}'s heteroskedasticity-consistent standard errors in their chosen models.  We denote these by $\hat{p}^{\hc 1}, \hat{p}^{\hc 2}$ and $\hat{p}^{\hc 3}$, respectively. 

In the full unmatched dataset, we would only expect the p-value based on the saturated model to perform well.  After all, the baseline model is misspecified and the model selector's p-value is rendered irregular by model selection.  However, the story is entirely different in the matched sample.

\begin{theorem}[Matching reduces model dependence] \label{theorem:model_dependence}
Consider the setting above.  If unweighted regressions based on pair matching are used, assume that $P$ satisfies the conditions of Theorem 1.  If weighted regressions based on matching with replacement are used, assume only that $P$ satisfies the conditions of Theorem 2.  Let $\phi^{(k)} = \mathbb{I} \{ \hat{p}^{\hc k} < \alpha \}$ denote the accept/reject decision based on $\hat{p}^{\hc k}$.  Then we have:
\begin{align*}
\lim_{n \rightarrow \infty} P_{h_n,g}^n ( \phi^{(1)} = \phi^{(2)} = \phi^{(3)}) = 1.
\end{align*}
In words, all three models yield the same conclusion with high probability.
\end{theorem}

This phenomenon may be understood as follows.  Under the assumptions of \Cref{theorem:model_dependence}, matching is able to balance any function of $X$ (i.e. $\sum_{Z_i = 1} \{ g(X_i) - g(X_{m(i)}) \} / N_1 \rightarrow 0$).  This makes $Z$ and $g(X)$ approximately orthogonal in the matched sample.  From standard least-squares theory, we know that the inclusion or exclusion of a nearly orthogonal predictor has very little impact on the other regression coefficients, explaining the similarity of the three p-values.  Note that this argument would not work in the full data, since there is no reason to expect $Z$ and $g(X)$ to be nearly orthogonal before matching.

\section*{Acknowledgement}

We are grateful to several anonymous referees and the associate editor for valuable comments and suggestions.  We thank John Cherian, Peng Ding, Colin Fogarty, Isaac Gibbs, Samir Khan, Sam Pimentel, Fredrik S\"avje, and seminar participants at various conferences for feedback.  We especially thank Sky Cao for help with the proof of Lemma B.4.

\bibliography{matchbib.bib}
\bibliographystyle{plain}

\newpage 
\appendix
\section{Proofs of main results} \label{section:proofs}

In this appendix, we prove the claims made in the main body of the paper. 

\subsection{Notation}

Throughout, we will use the following notations.
\begin{itemize}
\item \emph{Linear algebra}.  For any integer $k \geq 1$, $e_k$ is the $k$-th standard basis vector (where the ambient dimension will be clear from context) and $[k]$ is the set $\{ 1, \ldots, k \}$;  for two symmetric matrices $\mathbf{A}, \mathbf{B}$, we write $\mathbf{A} \succeq \mathbf{B}$ (resp. $\mathbf{A} \succ \mathbf{B}$) if $\mathbf{A} - \mathbf{B}$ is positive semidefinite (resp. positive definite).  We let $\lambda_{\min}(\mathbf{A})$ and $\lambda_{\max}(\mathbf{A})$ denote the smallest and largest eigenvalues of $\mathbf{A}$, respectively.  All vectors are interpreted as column vectors.  The concatenation of two vectors $v, u$ is denoted by $(v, u)$, which is still a column vector.  

\item \emph{Subvectors}.  For a vector $h \in R^k$ and a set $S \subseteq [k]$, $h_S \in \R^{|S|}$ is the subvector corresponding to entries in $S$.  For a function $g : \R^d \rightarrow \R^k$, $g_S : \R^d \rightarrow \R^{|S|}$ denotes the function $x \mapsto g(x)_S$.  When $S = \emptyset$, we abuse notation and let $h_S = 0 \in \R$ and $g_S \equiv 0$.

\item \emph{Weak convergence}.  Given a sequence of random distributions $\hat{Q}_n$, we say that $\hat{Q}_n$ converges \emph{weakly in probability} to $Q$ (denoted $\hat{Q}_n \rightsquigarrow_* Q$) if $\rho(\hat{Q}_n, Q) = o_P(1)$ for some distance $\rho$ metrizing weak convergence.  Often, we will apply this with $\hat{Q}_n$ the distribution of a statistic $H_*$ conditional on the original data $\mathcal{D}_n = \{ (X_i, Y_i, Z_i) \}_{i \leq n}$.  In such cases, we will also write $H_* \rightsquigarrow_* H$ where $H \sim Q$.

\item \emph{Conditional means and variances}.   We use $\mu_z(x)$ to denote the regression function $\E(Y \mid X = x, Z = z)$ and $\sigma_z^2(X)$ to denote the conditional variance $\Var(Y \mid X = x, Z = z)$.  If $P$ is also assumed to satisfy Fisher's null hypothesis, then we will typically drop the subscript $z$ and simply write $\mu(x), \sigma^2(x)$ since the conditional distribution of $Y$ given $(X, Z)$ will not depend on $Z$.  We also set $p = P(Z = 1)$.

\item \emph{Randomization critical values}.  For a statistic $\hat{\tau}_* \equiv \hat{\tau}( \{ (X_i, Y_i, Z_i^*) \}_{i \in \M})$, we define the randomization critical value $\hat{q}_{1 - \alpha}$ by $\hat{q}_{1 - \alpha} = \inf \{ t \in \R \, : \, \P( | \hat{\tau}_*| \leq t \mid \mathcal{D}_n) \geq 1 - \alpha \}$, where $\mathcal{D}_n = \{ (X_i, Y_i, Z_i) \}_{i \leq n}$ is the original data.  In the special cases of the difference-of-means and regression adjusted test statistics $\hat{\tau}^{\dm}, \hat{\tau}^{\reg}$, we denote the randomization quantiles by $\hat{q}_{1 - \alpha}^{\dm}$ and $\hat{q}_{1 - \alpha}^{\reg}$, respectively.
\end{itemize}

\subsection{Proof of \Cref{example:one_covariate}} \label{section:one_covariate_proof}

\begin{proof}
Let $\hat{\Delta} = \sum_{Z_i = 1} \{ X_i - X_{m(i)} \} / N_1$.  We begin by showing that $\hat{\Delta} \xrightarrow{P} \Delta = \{ 2(\theta_0 + \theta_1) - 1 \}^3 / 3 \theta_1^2(2 \theta_0 + \theta_1) \neq 0$.  Essentially, this follows from \cite[Proposition 1]{savje2021}, which shows that $\E( \sum_{Z_i = 1} \{ X_i - X_{m(i)} \} / np) \rightarrow \Delta$.  However, a small extension of his analysis is needed to obtain convergence in probability.

We briefly sketch this extension, freely using the notation from the Supplement of \cite{savje2021}.  Begin by decomposing the covariate imbalance into the sum of several terms:
\begin{align*}
\hat{\Delta} &= \frac{1}{np} \sum_{Z_i = 1} \{ X_i - X_{m(i)} \} + o_P(1)\\
&= \underbrace{\frac{1}{np} \sum_{i \in \mathcal{C}_+} X_i - \frac{1}{np} \sum_{i \in \M_+^*} X_i}_a + \underbrace{\frac{1}{np} \sum_{i \in \mathcal{T}_-} X_i - \frac{1}{np} \sum_{i \in \M^*_-} X_i}_b + \underbrace{\frac{1}{np} \sum_{i \in \mathcal{T}_+} X_i - \frac{1}{np} \sum_{i \in \mathcal{C}_+} X_i}_c + o_P(1)
\end{align*}
The proof of \cite[Lemma S7]{savje2021} shows that $\E(|a|) \leq  2\E( | \mathcal{C}_+ \backslash \M_+^*| /np) + 2 \E\{( | \M_+^*| - | \mathcal{C}_+|) / np \}$, and the two terms in this upper bound are shown to vanish in Lemmas S8 and S11 of \cite{savje2021}, respectively.  Hence, $a = o_P(1)$ by Markov's inequality.  Similarly, the proof of \cite[Lemma S12]{savje2021} shows that $b \leq \sum_{i \in \mathcal{T}_e} | X_i - X_{m_e^*(i)}| / np$, and Lemmas S8, S11, S12 and S13 in \cite{savje2021} show that the expectation of this upper bound vanishes.  Hence, $b = o_P(1)$ as well.  Finally, term $c$ converges to $\Delta$ by the law of large numbers.  Since $P \{ e(X) \geq 0.5 \} > 0$, we must have $\theta_0 + \theta_1 \neq 0.5$ and so $\Delta \neq 0$.

Now, we are ready to show $\P( \hat{p}^{\dm} < \alpha) \rightarrow 1$ for every $\alpha \in (0, 1)$.  Let $\epsilon_i = Y_i - \beta_0 - \beta_1 X_i$.
Since $\hat{\tau}^{\dm} = \beta_1 \hat{\Delta} + \sum_{Z_i = 1} \{ \epsilon_i - \epsilon_{m(i)} \} / N_1 = \beta_1 \hat{\Delta} + o_P(1)$, the continuous mapping theorem implies $| \hat{\tau}^{\dm}|$ converges in probability to $| \beta_1 \Delta| > 0$.  Write $\P( \hat{p}^{\dm} < \alpha) = \P( | \hat{\tau}^{\dm}| \geq \hat{q}_{1 - \alpha}^{\dm})$.  Since $| \hat{\tau}^{\dm}|$ tends to a positive constant but $\hat{q}_{1 - \alpha}^{\dm}$ tends to zero at rate $O_P(n^{-1/2})$ (\Cref{lemma:dm_critical_value}), the chance that $| \hat{\tau}^{\dm}|$ exceeds $\hat{q}_{1 - \alpha}^{\dm}$ tends to one.
\end{proof}

\subsection{Proof of \Cref{example:two_covariates}} \label{section:two_covariates_proof}

\begin{proof}
It suffices to show that $\hat{\Delta} = \sum_{Z_i = 1} \{ X_i - X_{m(i)} \} / N_1$ satisfies $\P(\| \hat{\Delta} \| > \eta) \rightarrow 1$ for some $\eta > 0$.  From there, the rest of the proof goes through analogously to the proof of \Cref{example:one_covariate}.  To accomplish this, we will show that $\theta^{\top} \hat{\Delta} = \sum_{Z_i = 1} \{ \mu(X_i) - \mu(X_{m(i)}) \} / N_1$ does not vanish.  Throughout, we will make use of the conditional distribution of $\mu(X)$:
\begin{align*}
f_{\mu}(h \mid Z = z) = \frac{2}{\pi} \sqrt{1 - h^2} \times \left\{ \begin{array}{ll}
1 + h &\text{if } z = 1\\
1 - 7h/13 &\text{if } z = 0
\end{array}
\right.
\end{align*}

The differnece $\theta^{\top} \hat{\Delta}$ can be split into two averages, $\sum_{Z_i = 1} \mu(X_i)/N_1$ and $\sum_{Z = 1} \mu(X_{m(i)})/ N_1$.  The former converges to $\int_{-1}^1 h f_{\mu}(h \mid 1) \, \d h = 0.25$ by the law of large numbers.  We will show that the latter average is asymptotically bounded away from $0.25$.  We do this by showing that, with high probability, the average of even the $N_1$ \emph{largest} values of $\mu(X_i)$ among untreated observations $i$ is still smaller than $0.25$.  Thus, no pair matching scheme will be able to balance $\mu(\cdot)$.  This is roughly the approach taken in \cite{rubin1976, rubin_thomas2000}.

Let $\mathcal{H}$ collect the indices of the $N_1$ untreated observations with the largest values of $\mu(X_i)$, with $\mathcal{H} = \emptyset$ when there are fewer than $N_1$ untreated observations.  We claim that, with high probability, $\mathcal{H}$ includes all untreated values of $\mu(X_i)$ larger than $-0.22$.  This follows from the following calculation:
\begin{align*}
\frac{1}{N_1} \sum_{Z_i = 0} \mathbf{1} \{ \mu(X_i) \geq -0.22 \} &= \frac{\sum_{Z_i = 0} \mathbf{1} \{ \theta^{\top} X_i \geq -0.22 \} / N_0}{N_1 / N_0} \xrightarrow{P} \frac{P \{ \mu(X) \geq -0.22 \mid Z = 0 \}}{0.35 / 0.65} < \frac{0.533}{0.538} < 1
\end{align*}
Thus, with probability approaching one $\sum_{Z_i = 0} \mathbf{1} \{ \mu(X_i) \geq -0.22 \} < N_1$ and so $\mathcal{H}$ must contain all untreated units with $\mu(X_i) \geq -0.22$.  The other units in $\mathcal{H}$ must have $\mu(X_i) < -0.22$, so we may write:
\begin{align*}
\frac{1}{N_1} \sum_{Z_i = 1} \mu(X_{m(i)}) &\leq \frac{1}{N_1} \sum_{i \in \mathcal{H}} \mu(X_i)\\
&< \frac{1}{N_1} \sum_{Z_i = 0} \mu(X_i) \mathbf{1} \{ \mu(X_i) \geq -0.22 \} \\
&= \frac{\sum_{i = 1}^n (1 - Z_i) \mu(X_i) \mathbf{1} \{ \mu(X_i) \geq -0.22 \} / n}{N_1 / n}\\
&\xrightarrow{P} \frac{0.65 \E( \mu(X) \mathbf{1} \{ \mu(X) \geq -0.22 \} \mid Z = 0)}{0.35}\\
&\approx 0.2386.
\end{align*}
Thus, with probability approaching one $\sum_{Z_i = 1} \mu(X_{m(i)}) / N_1 < 0.24$ while $\sum_{Z_i = 1} \mu(X_i)/N_1 > 0.245$.  This proves that $\| \Delta \|$ is asymptotically bounded away from zero.
\end{proof}

\subsection{Proof of \Cref{example:regression_adjusted_test_statistics}} \label{section:regression_adjusted_test_statistics_proof}

\begin{proof}
Let $\mathbf{B} = \sum_{i \in \M} (Z_i - 0.5, 1, X_i)(Z_i - 0.5, 1, X_i)^{\top} / 2 N_1$ be the matched design matrix.  \Cref{lemma:B_invertibility} (applied with $\phi(x) = x$) shows that $\mathbf{B}^{-1}$ exists with probability approaching one.  On this event, standard least-squares theory gives $\sqrt{N_1} \hat{\tau}^{\reg} \mid \{ (X_i, Z_i) \}_{i \leq n} \sim N \{ 0, 0.5 \sigma^2 (\mathbf{B}^{-1})_{11} \}$.

Let $\Delta$ be the asymptotic matching discrepancy introduced in \Cref{section:one_covariate_proof}, and set $\eta = 1 - 1/\{1 - 2 \Delta^2 / \E(X^2) \} > 0$.  We now show that with probability approaching one, the following occurs:
\begin{align}
0.5 \sigma^2 (\mathbf{B}^{-1})_{11} \geq 2 \sigma^2 (1 + \eta). \label{eq:variance_inflation}
\end{align}
To prove this, define $\hat{\mathbf{\Sigma}}_X = \sum_{i \in \M} X_i^2 / 2 N_1$ and observe that $\hat{\mathbf{\Sigma}}_X \leq \sum_{i = 1}^n X_i^2 / 2 N_1 \xrightarrow{P} \E(X^2) / 2p$.  Thus, with probability approaching one, $\hat{\mathbf{\Sigma}}_X \leq \E(X^2) / 2$.  When this occurs, we may use the partitioned matrix inversion formula to conclude $(\mathbf{B}^{-1})_{11} = 4/(1 - \hat{\Delta}^2 / \hat{\mathbf{\Sigma}}_X) \geq 4(1 + \eta)$.  This proves \Cref{eq:variance_inflation}.

Now, let $\epsilon > 0$ be arbitrary.  \Cref{lemma:paired_ancova_critical_value} shows that $\sqrt{N_1} \hat{q}_{1 - \alpha}^{\reg} \rightarrow \sqrt{2} \sigma z_{1 - \alpha/2}$ in probability, so (\ref{eq:critical_value_inequality}) holds with probability approaching one:
\begin{align}
\sqrt{N_1} \hat{q}_{1 - \alpha}^{\reg} \geq \sqrt{2} \sigma z_{1 - \alpha/2} - \epsilon. \label{eq:critical_value_inequality}
\end{align}

When (\ref{eq:variance_inflation}) and (\ref{eq:critical_value_inequality}) both hold, we may write:
\begin{align*}
\P( \hat{p}^{\reg} < \alpha \mid \{ (X_i, Z_i) \}_{i \leq n}) &= \P( | \sqrt{N_1} \hat{\tau}^{\reg}| \geq \hat{q}_{1 - \alpha}^{\reg} \mid \{ (X_i, Z_i) \}_{i \leq n})\\
&\geq \P( | \sqrt{N_1} \hat{\tau}^{\reg} | \geq \sqrt{2} \sigma z_{1 - \alpha/2} - \epsilon  \mid \{ (X_i, Z_i) \}_{i \leq n}) &\text{(\ref{eq:critical_value_inequality})}\\
&\geq \P_{H \sim N \{ 0, 2 \sigma^2(1 + \eta) \}}(|H| \geq \sqrt{2} \sigma z_{1 - \alpha/2} - \epsilon) &\text{(\ref{eq:variance_inflation}) + Anderson's Lemma}.
\end{align*}
Since (\ref{eq:variance_inflation}) and (\ref{eq:critical_value_inequality}) both hold with probability tending to one, we may take expectations on both sides and conclude $\limsup \P( \hat{p}^{\reg} < \alpha) \geq \P_{H \sim N \{ 0, 2 \sigma^2(1 + \eta) \}}(|H| \geq \sqrt{2} \sigma z_{1 - \alpha/2} - \epsilon)$.  For small enough $\epsilon$, this lower bound is strictly larger than $\alpha$.
\end{proof}

\subsection{Proof of \Cref{proposition:balance_requirement}}

\begin{proof}
First, we prove sufficiency.  Let $\F_n = \sigma( \{ (X_i, Z_i) \}_{i \leq n})$.  Suppose that $\E( \hat{\tau}^{\dm} \mid \F_n) = o_P(n^{-1/2})$.  Let $s^2 = 2 \E \{ \sigma^2(X) \mid Z = 1 \}$.  Then we may write:
\begin{align*}
\sqrt{N_1} \hat{\tau}^{\dm} &= \sqrt{N_1} \E( \hat{\tau}^{\dm} \mid \F_n) + \frac{1}{\sqrt{N_1}} \sum_{Z_i = 1} ( \epsilon_i - \epsilon_{m(i)} ) = o_P(1) + \frac{1}{\sqrt{N_1}} \sum_{Z_i = 1} ( \epsilon_i - \epsilon_{m(i)} ) \rightsquigarrow N(0, s^2).
\end{align*}
where the final convergence follows from the Berry-Esseen Theorem applied conditional on $\F_n$ (see \Cref{lemma:epsilon_difference_asymptotic_normality} for a complete proof).  \Cref{lemma:dm_critical_value} studies the randomization critical value and proves $\sqrt{N_1} \hat{q}_{1 - \alpha}^{\dm} =  s z_{1 - \alpha/2} + o_P(1)$.  Thus, Slutsky's theorem gives $\P( \hat{p}^{\dm} < \alpha) = \P(| \sqrt{N_1} \hat{\tau}^{\dm}| \geq \hat{q}_{1 - \alpha}^{\dm}) \rightarrow \P_{H \sim N(0, s^2)}(|H| \geq s z_{1 -\alpha/2}) = \alpha$.  We have now shown that small conditional bias implies asymptotic validity.

Next, we prove necessity.  Let $s_n^2 = \sum_{i \in \M} \sigma^2(X_i) / N_1$.  \Cref{lemma:sigma_star_limit} shows that $s_n / s \rightarrow 1$ in probability when $P \{ e(X) < 0.5 \} = 1$; in particular, $1/s_n = O_P(1)$.  Thus, $\E( \hat{\tau}^{\dm} \mid \F_n) \neq o_P(n^{-1/2})$ implies $B_n := \E( \sqrt{N_1} \hat{\tau}^{\dm} / s_n \mid \F_n) \neq o_P(1)$ as well, meaning that there exists some $\eta > 0$ such that the following holds:
\begin{align*}
\liminf_{n \rightarrow \infty} \P(|B_n| > \eta) > \eta.
\end{align*}
By passing to a subsequence if necessary, we may assume the above holds for all $n$ and not just asymptotically.  

As explained above, \Cref{lemma:dm_critical_value} shows that $\sqrt{N_1} \hat{q}_{1 - \alpha}^{\dm} = s z_{1 - \alpha/2} + o_P(1)$.  Thus, for any $\epsilon > 0$, we have $\sqrt{N_1} \hat{q}_{1 - \alpha}^{\dm} / s_n > z_{1 - \alpha/2} - \epsilon$ with probability approaching one.  This allows us to write the following:
\begin{align*}
\liminf_{n \rightarrow \infty} \P( \hat{p}^{\dm} < \alpha) &= \liminf_{n \rightarrow \infty} \P( | \hat{\tau}^{\dm}| \geq \hat{q}_{1 - \alpha}^{\dm})\\
&= \liminf_{n \rightarrow \infty} \P( | \sqrt{N_1} \hat{\tau}^{\dm} / s_n| \geq \sqrt{N_1} \hat{q}_{1 - \alpha}^{\dm}, \sqrt{N_1} \hat{q}_{1 - \alpha}^{\dm} / s_n > z_{1 - \alpha/2} - \epsilon)\\
&+ \liminf_{n \rightarrow \infty} \P( | \sqrt{N_1} \hat{\tau}^{\dm} / s_n| \geq \sqrt{N_1} \hat{q}_{1 - \alpha}^{\dm}, \sqrt{N_1} \hat{q}_{1 - \alpha}^{\dm} / s_n \leq z_{1 - \alpha/2} - \epsilon)\\
&\geq \liminf_{n \rightarrow \infty} \P( | \sqrt{N_1} \hat{\tau}^{\dm} / s_n| > z_{1 - \alpha/2} - \epsilon) + 0.
\end{align*}
Based on this lower bound, it suffices to show $\liminf \P( | \sqrt{N_1} \hat{\tau}^{\dm} / s_n| > z_{1 - \alpha/2} - \epsilon) > \alpha$ for some $\epsilon > 0$.  

We show this by conditioning on $\F_n$:
\begin{align*}
\P( | \sqrt{N_1} \hat{\tau}^{\dm} / s_n| > z_{1 - \alpha/2} - \epsilon) &= \E\{ \P( | \sqrt{N_1} \hat{\tau}^{\dm} / s_n| > z_{1 - \alpha/2} - \epsilon \mid \F_n) \}\\
&= \E \left[ \P \left\{ \left| B_n  + \frac{1}{s_n \sqrt{N_1}} \sum_{Z_i = 1} ( \epsilon_i - \epsilon_{m(i)}) \right| > z_{1 - \alpha/2} - \epsilon \, \bigg| \, \F_n \right\} \right]\\
&\geq \underbrace{\E \left[ \P_{H \sim N(0, 1)}\{ |B_n + H| > z_{1 - \alpha/2} - \epsilon \mid \F_n \} \right]}_a\\
&- \underbrace{\E \left[ 2\sup_{t \in \R} \left| \P \left\{ \frac{1}{s_n \sqrt{N_1}} \sum_{Z_i = 1} (\epsilon_i - \epsilon_{m(i)}) \leq t  \, \bigg| \, \F_n \right\} - \P_{H \sim N(0, 1)}(H \leq t) \right| \right]}_b
\end{align*}
Term $b$ tends to zero by (\ref{eq:epsilon_difference_kolmogorov_convergence}) and the bounded convergence theorem.  Meanwhile, term $a$ can be further controlled using the fact that $x \mapsto \P_{H \sim N(0, 1)}(|x + H| > z_{1 - \alpha/2} - \epsilon)$ is increasing as $x$ moves away from zero:
\begin{align*}
a &= \E[ \P ( | B_n + H| > z_{1 - \alpha/2} - \epsilon ) \mathbf{1} \{ |B_n| > \eta \}] + \E[ \P( |B_n + H| > z_{1 - \alpha/2} - \epsilon) \mathbf{1} \{ |B_n| \leq \eta \}]\\
&\geq \E[ \P(| \eta + H| > z_{1 -\alpha/2} - \epsilon) \mathbf{1} \{ |B_n| > \eta \}] + \E[ \P( |H| > z_{1 - \alpha/2} - \epsilon) \mathbf{1} \{ |B_n| \leq \eta \}]\\
&= \P(| \eta + H| > z_{1 - \alpha/2} - \epsilon) \P(|B_n| > \eta) + \P(|H| > z_{1 - \alpha/2} - \epsilon) \P(|B_n| \leq \eta)\\
&\geq \P(| \eta + H| > z_{1 - \alpha/2} - \epsilon) \eta + \P(|H| > z_{1 - \alpha/2} - \epsilon)(1 - \eta)\\
&\geq \P(|H| > z_{1 - \alpha/2} - \epsilon) + \eta \{ \P(|\eta + H| > z_{1 - \alpha/2} - \epsilon) - \P(|H| > z_{1 - \alpha/2} - \epsilon) \}
\end{align*}
As $\epsilon \rightarrow 0$, the lower bound in the preceding display converges to $\alpha + \eta \{ \P(| \eta + H| > z_{1 - \alpha/2}) - \alpha \} > \alpha$.  Hence, for small enough $\epsilon$, we have $a > \alpha$ and we have shown $\liminf \P( \hat{p}^{\dm} < \alpha) > \alpha$.
\end{proof}

\subsection{Proof of \Cref{proposition:dm_sufficient}} \label{section:dm_sufficient_proof}

\begin{proof}
Under the stated assumptions, \cite[Proposition 1]{abadie_imbens2012} (or more properly, its proof) implies that $\sum_{Z_i = 1} | X_i - X_{m(i)} | / N_1 = o_P(n^{-1/2})$.  By the Lipschitz condition on the outcome model, this also yields $\E( \hat{\tau}^{\dm} \mid \{ (X_i, Z_i) \}_{i \leq n}) = o_P(n^{-1/2})$.  Now validity follows from \Cref{proposition:balance_requirement}.
\end{proof}

\subsection{Proof of \Cref{proposition:ols_sufficient}} \label{section:ols_sufficient_proof}

\begin{proof}
Let $s^2 = 2 \E \{ \sigma^2(X) \mid Z = 1 \}$.  \Cref{lemma:paired_expansion} (applied with $g \equiv 0$) shows that $\sqrt{N_1} \hat{\tau}^{\reg} \rightsquigarrow N(0, s^2)$.  Meanwhile, \Cref{lemma:paired_ancova_critical_value} shows that $\sqrt{N_1} \hat{q}_{1 - \alpha}^{\reg} \rightarrow s z_{1 - \alpha/2}$ in probability.  Therefore, Slutsky's theorem gives $\P( \hat{p}^{\reg} < \alpha) = \P( | \sqrt{N_1} \hat{\tau}^{\reg}| \geq \hat{q}_{1 - \alpha}^{\reg}) \rightarrow \P_{H \sim N(0, s^2)}(|H| \geq s z_{1 -\alpha/2}) = \alpha$.
\end{proof}

\subsection{Proof of \Cref{theorem:paired_local_robustness}}

\begin{proof}
For simplicity, we only prove the result for one-sided p-values.  Let $g : \R^d \rightarrow \R^k$ be any bounded nonlinear function and $\{ h_n \} \subset \R^k$ a sequence with $\| h_n \| \leq C n^{-1/2}$.  Let $\psi(x, z) = (z - 0.5, 1, x)$ be the ``feature vector" used in the regression that computes $\hat{\tau}^{\reg}$.  Note we have shifted the treatment variable by a constant, although this does not affect the value of $\hat{\tau}^{\reg}$ since an intercept is present.

Define the following statistics:
\begin{align*}
(\tilde \tau_n^{\reg}, \tilde \theta) &= \argmin_{\tau, \theta} \sum_{i \in \M} \{ Y_i + h_n^{\top} g(X_i) - (\tau, \theta)^{\top} \psi(X_i, Z_i) \}^2\\
\tilde \epsilon_i &= Y_i + h_n^{\top} g(X_i) - (\tilde \tau_n^{\reg}, \tilde \theta)^{\top} \psi(X_i, Z_i)  \\
\tilde \sigma^2_{\hc} &= e_1^{\top} \left( \sum_{i \in \M} \psi(X_i, Z_i) \psi(X_i, Z_i)^{\top} \right)^{-1} \left( \sum_{i \in \M} \tilde \epsilon_i^2 \psi(X_i, Z_i) \psi(X_i, Z_i)^{\top} \right) \left( \sum_{i \in \M} \psi(X_i, Z_i) \psi(X_i, Z_i)^{\top} \right) e_1.
\end{align*}
In words, these are the regression coefficients, fitted residuals, design matrix, and heteroskedasticity-consistent robust standard error in the regression $Y + h_n^{\top} g(X) \sim 1 + (Z - 0.5) + X$.  

The reason to consider these statistics is the following: the distribution of the $t$-statistic $\tilde \tau_n^{\reg} / \tilde \sigma_{\hc}$ under $P^n$ is exactly the same as the distribution of the $t$-statistic $\hat{\tau}_n^{\reg} / \hat{\sigma}_{\hc}$ under $P_{h_n,g}^n$, where $\hat{\sigma}_{\hc}$ is the robust standard error for $\hat{\tau}_{\reg}$ in the regression $Y \sim 1 + (Z - 0.5) + X$.  In particular, we have:
\begin{align*}
P_{h_n,g}^n( \hat{p}^{\hc} < \alpha) &= P_{h_n,g}^n( \hat{\tau}_n^{\reg} / \hat{\sigma}_{\hc} > z_{1 - \alpha}) = P^n( \tilde \tau_n^{\reg} / \tilde{\sigma}_{\hc} > z_{1 - \alpha})
\end{align*}
and it remains to show $P^n( \tilde \tau_n^{\reg} / \tilde \sigma_{\hc} > z_{1 - \alpha}) \rightarrow \alpha$.

To prove this, we will use the fact that $\sqrt{N_1} \tilde \tau_n^{\reg} \rightsquigarrow N(0, s^2)$ even though the linear model is misspecified, where $s^2 = 2 \E \{ \sigma^2(X) \mid Z = 1 \}$.  This is proved formally in \Cref{lemma:paired_expansion} (applied with $S = \emptyset$), although we will sketch the intuition here.  Since $P \{ e(X) < 0.5 \} = 1$, optimal matching will succeed at balancing all functions of $X$ as the sample size grows large.  More formally, \Cref{corollary:empirical_Lp_convergence} shows that $\sum_{Z_i = 1} \{ g(X_i) - g(X_{m(i)}) \} / N_1 \rightarrow 0$ in probability, as long as $\Var \{ g(X) \} < \infty$.  This balancing property makes the $Z_i - 0.5$ approximately orthogonal to the nonlinearity $h_n^{\top} g(X_i)$, in the sense that their sample covariance tends to zero even after multiplying by $\sqrt{N_1}$:
\begin{align*}
\sqrt{N_1} \times \frac{1}{2N_1} \sum_{i \in \M} (Z_i - 0.5) h_n^{\top} g(X_i) &= \underbrace{\frac{\sqrt{N_1} h_n^{\top}}{4}}_{= O_P(1)} \times \underbrace{\frac{1}{N_1} \sum_{Z_i = 1} \{ g(X_i) - g(X_{m(i)}) \}}_{= o_P(1)} = o_P(1).
\end{align*}
This orthogonality makes it so that the omission of $g$ in the regression has very little effect on $\tilde \tau_n^{\reg}$.

The second fact we use is that $\sqrt{N_1} \tilde \sigma_{\hc} \xrightarrow{P} s$.  This is to be expected.  For example, if we assumed that the distribution of $Y \mid X, Z$ were suitably smooth, then it would follow directly from contiguity and the usual consistency of robust standard errors \cite{white1980}.  Since we do not assume smoothness, a more ``bare hands" proof is given in \Cref{lemma:paired_hc_standard_error}.

Combining the two facts with Slutsky's theorem gives $P^n( \tilde \tau_n^{\reg} / \tilde \sigma_{\hc} > z_{1 - \alpha}) \rightarrow \alpha$.  This proves that model-based inference remains valid along any locally misspecified sequence.  Since the above argument holds for any sequence $h_n$ satisfying $\| h_n \| \leq C n^{-1/2}$, it must also be the case that $\sup_{\| h \| \leq C n^{-1/2}} P_{h,g}^n( \hat{p}^{\reg} < \alpha) \rightarrow \alpha$.
\end{proof}

\subsection{Proof of \Cref{theorem:replaced_robustness}}

\begin{proof}
The prove is identical to that of \Cref{theorem:paired_local_robustness}, except we get asymptotic normality from \Cref{lemma:replaced_expansion} instead of \Cref{lemma:paired_expansion}, and the asymptotics of the robust standard error come from \Cref{lemma:replaced_hc_standard_error} instead of \Cref{lemma:paired_hc_standard_error}.
\end{proof}

\subsection{Proof of \Cref{theorem:model_dependence}}

\begin{proof}
Consider the case of where $P$ satisfies the conditions of \Cref{theorem:paired_local_robustness} and optimal pair matching is used.  For any nonempty set $S \subseteq [k]$, let $\psi_S(x, z) = (z - 0.5, 1, x, g_S(x))$ and set $\psi_{\emptyset}(x, z) = (z -0.5, 1, x)$.  Further define the following quantities:
\begin{align*}
(\tilde \tau_S, \tilde \theta_S) &= \argmin_{\tau, \theta} \sum_{i \in \M} \{ Y_i + h_n^{\top} g(X_i) - (\tau, \theta)^{\top} \psi_S(X_i, Z_i) \}^2\\
\tilde \epsilon_{S,i} &= Y_i + h_n^{\top} g(X_i) - (\tilde \tau_S, \tilde \theta_S)^{\top} \psi_S(X_i, Z_i)\\
\tilde \sigma^2_S &= e_1^{\top} \left( \sum_{i \in \M} \psi_S(X_i, Z_i) \psi_S(X_i, Z_i)^{\top} \right)^{-1} \left( \sum_{i \in \M} \tilde \epsilon_{S,i}^2 \psi_S(X_i, Z_i) \psi_S(X_i, Z_i)^{\top} \right) \left( \sum_{i \in \M} \psi_S(X_i, Z_i) \psi_S(X_i, Z_i)^{\top} \right)^{-1} e_1\\
\tilde t_S &= \tilde \tau_S / \tilde \sigma_S 
\end{align*}
In words, these are the regression coefficients, residuals, and robust standard errors, and $t$-statistic from a regression of $Y_i + h_n^{\top} g(X_i)$ on $\psi_S(X_i, Z_i)$. 

Let $t_S$ be the $t$-statistic from a regression of $Y_i$ on $\psi_S(X_i, Z_i)$ (so that $t_S = \tilde t_S$ when $h_n \equiv 0$).  Then the joint distribution of $(t_S \, : \, S \subseteq [k])$ under $P_{h_n,g}^n$ is the same as the joint distribution of $(\tilde t_S \, : \, S \subseteq [k])$ under $P^n$.  In particular, we have:
\begin{align*}
P_{h_n,g}^n( \phi^{(1)} = \phi^{(2)} = \phi^{(3)}) &\geq P_{h_n,g}^n( \mathbb{I} \{ | t_S | > z_{1 - \alpha/2} \} \text{ is the same for all } S \subseteq [k])\\
&= P^n( \mathbb{I} \{ | \tilde t_S | > z_{1 - \alpha/2} \} \text{ is the same for all } S \subseteq [k]).
\end{align*}
To show that the last line in the preceding display tends to one, we use Lemmas \ref{lemma:paired_expansion} and \ref{lemma:paired_hc_standard_error}, which jointly show that the expansion $\tilde t_S = \sum_{i \in \M} (2 Z_i - 1) \epsilon_i / s \sqrt{N_1} + o_P(1)$ holds for every $S \subseteq [k]$, where $s^2 = 2 \E \{ \sigma^2(X) \mid Z = 1 \}$.  Therefore, we have:
\begin{align*}
(\tilde t_S \, : \, S \subseteq [k]) \overset{P^n}{\rightsquigarrow} (H_S \, : \, H \subseteq [k]) \quad \text{where } H_S \equiv H \sim N(0, 1) \text{ for all } S \subseteq [k].
\end{align*}
In other words, the $t$-statistics jointly converge to a degenerate vector with all components equal to the same standard normal random variable.

Let $C$ be the set of vectors in $\R^{2^k}$ with all entries either all strictly above $z_{1 - \alpha/2}$ or all weakly less than $z_{1 - \alpha/2}$.  Since $\P \{ (H_S \, : \, H \subseteq [k]) \in \partial C) = 0$, the Portmanteau lemma gives:
\begin{align*}
P^n( \mathbb{I} \{ | \tilde t_S | > z_{1 - \alpha/2} \} \text{ is the same for all } S \subseteq [k]) &= P^n\{ (\tilde t_S \, : \, S \subseteq [k]) \in C \}\\
&\rightarrow P \{ (H_S \, : \, S \subseteq [k]) \in C \}\\
&= 1.
\end{align*}
Thus, we have shown $P_{h_n,g}^n( \phi^{(1)} = \phi^{(2)} = \phi^{(3)}) \rightarrow 1$.

The proof in the case of matching with replacement is identical, except we use Lemmas \ref{lemma:replaced_expansion} and \ref{lemma:replaced_hc_standard_error} instead of Lemmas \ref{lemma:paired_expansion} and \ref{lemma:paired_hc_standard_error} to show that all $t$-statistics are asymptotically equivalent.
\end{proof}

\section{Supporting results} \label{section:technical}

In this section, we prove supporting technical results used in \Cref{section:proofs}.  Apart from \Cref{section:technical_lemmas}, we will always assume that $P$ satisfies \Cref{assumption:primitives}.

\subsection{Miscellaneous technical results} \label{section:technical_lemmas}

\begin{theorem}
\label{theorem:berry_esseen}
\emph{(Berry-Esseen Theorem, \cite{Petrov2000})}. Let $X_1, \ldots, X_n$ be independent but not necessarily identically distributed random variables with mean zero and three finite moments.  Let $S_n = \sum_{i = 1}^n X_i$ and $\sigma^2 = \Var(S_n)$.  Then $\sup_{t \in \R} | \P(S_n/\sigma \leq t) - \Phi(t)| \leq C \sum_{i = 1}^n \E( |X_i|^3) / \sigma^3$ for some absolute constant $C < \infty$.
\end{theorem}

\begin{lemma}[Randomization Slutsky Theorem]
\label{lemma:randomization_slutsky}
If $A_n \rightsquigarrow_* A$ and $R_n = o_P(1)$, then $A_n + R_n \rightsquigarrow_* A$ as well.
\end{lemma}

\begin{proof}
For any subsequence $\{ n_k \}$ we may find a further sub-subsequence along which $A_n \overset{\text{a.s.}}{\rightsquigarrow} A$ and $R_n \xrightarrow{\text{a.s.}} 0$.  By the ordinary Slutsky Theorem, $A_n + R_n \overset{\text{a.s.}}{\rightsquigarrow} A$ along the sub-subsequence.  Since $\{ n_k \}$ is arbitrary, $A_n + R_n \rightsquigarrow_* A$ along the full sequence $n = 1, 2, 3, \ldots$.
\end{proof}

\begin{lemma}[Conditional WLLN]
\label{lemma:conditional_wlln} 
Let $A_n$ be a sequence of random variables with and $\F_n$ a sequence of $\sigma$-algebras.  Suppose that $\E(A_n \mid \F_n) = 0$ and $\Var(A_n \mid \F_n) = o_P(1)$.  Then $A_n = o_P(1)$.
\end{lemma}

\begin{proof}
For any $\epsilon > 0$, Chebyshev's inequality implies $\P( |A_n| > \epsilon \mid \F_n) \leq \Var( A_n \mid \F_n) / \epsilon^2$.  We also have the trivial bound $\P( |A_n| > \epsilon \mid \F_n) \leq 1$.  Thus, $B_n = \P( |A_n| > \epsilon \mid \F_n) \leq \min \{ 1, \Var( A_n \mid \F_n ) /\epsilon^2 \}$ so $B_n$ is a uniformly bounded sequence of random variables that tends to zero in probability.  By the bounded convergence theorem, $\P( |A_n| > \epsilon) = \E[ B_n ] \rightarrow 0$.
\end{proof}

\subsection{Theoretical background on optimal matching} \label{section:optimal_matching_background}

\begin{lemma} \label{lemma:mahalanobis_euclidean_equivalence}
Suppose that $P$ satisfies \Cref{assumption:primitives} and let $d_M(x, x') =  \{ (x - x')^{\top} \hat{\mathbf{\Sigma}}^{-1}(x - x') \}^{1/2}$ denote the Mahalanobis distance between $x, x'$ based on the estimated sample covariance matrix.  Then there exists a constant $\kappa(\mathbf{\Sigma}) < \infty$ such that $\P( \kappa(\mathbf{\Sigma})^{-1} \| x - x' \| \leq d_M(x, x') \leq \kappa(\mathbf{\Sigma}) \| x - x' \| \text{ for all } x, x' \in \R^d) \rightarrow 1$.
\end{lemma}

\begin{proof}
Observe that $\tfrac{1}{2} \lambda_{\min}( \mathbf{\Sigma}^{-1}) \leq \lambda_{\min}( \hat{\mathbf{\Sigma}}^{-1}) \leq \lambda_{\max}( \hat{\mathbf{\Sigma}}^{-1}) \leq 2 \lambda_{\max}( \mathbf{\Sigma}^{-1})$ with probability approaching one, by the law of large numbers and continuous mapping theorem.  When these inequalities hold, we have the following bounds for all $x, x'$:
\begin{align*}
d_M(x, x')^2 &= (x - x')^{\top} \hat{\mathbf{\Sigma}}^{-1}(x - x') \geq \lambda_{\min}( \hat{\mathbf{\Sigma}}^{-1}) \| x - x' \|^2 \geq \tfrac{1}{2} \lambda_{\min}( \mathbf{\Sigma}^{-1}) \| x - x' \|^2\\
d_M(x, x')^2 &= (x - x')^{\top} \hat{\mathbf{\Sigma}}^{-1}(x - x') \leq \lambda_{\max}( \hat{\mathbf{\Sigma}}^{-1}) \| x - x' \|^2 \leq 2 \lambda_{\max}( \mathbf{\Sigma}^{-1}) \| x - x' \|^2.
\end{align*}
Hence, the result holds with $\kappa(\mathbf{\Sigma}) = [\max \{ 2 \lambda_{\max}(\mathbf{\Sigma}^{-1}), 2 /\lambda_{\min}( \mathbf{\Sigma}^{-1}) \}]^{1/2}$.
\end{proof}

\begin{proposition}
\label{proposition:mahalanobis_discrepancy_vanishes}
If $P \{ e(X) < 0.5 \} = 1$, then $\sum_{Z_i = 1} \| X_i - X_{m(i)} \| / N_1 \xrightarrow{P} 0$.
\end{proposition}

\begin{proof}
First, we prove that for any $\epsilon > 0$, there exists a pair matching scheme $\bar{m}_{\epsilon}$ with $P( \sum_{Z_i = 1} \| X_i - X_{\bar{m}_{\epsilon}(i)} \| / N_1 > \epsilon) \rightarrow 0$.  We use a similar approach to the proof of \cite[Proposition 1]{abadie_imbens2012}.  Let $B < \infty$ be a number so large that $P(X \not \in [-B, B]^d) \leq p^2 \epsilon / 128 \E(\|X\|^2)$.  Then, divide $[-B, B]^d$ into $N$ disjoint cubes $C_1, \ldots, C_N$ of side length at most $\epsilon / 4 \sqrt{d}$.  

Now, let us define our matching $\bar{m}_{\epsilon}$.  Suppose that the following events both occur:
\begin{align}
&\left\{ \frac{1}{n} \sum_{i = 1}^n Z_i > p/2, \frac{1}{n} \sum_{i = 1}^n \| X_i \|^2 < 2 \E(\|X \|^2), \frac{1}{n} \sum_{i = 1}^n \mathbf{1} \{ X_i \not \in [-B, B]^d \} < 2 P( X \not \in [-B, B]^d), N_0 > N_1 \right\} \label{eq:good_event1}\\
&\{ \text{No cube $C_j$ contains more treated than untreated units} \} \label{eq:good_event2}
\end{align}
Then it is possible to pair each treated unit $i$ with $\| X_i \|_{\infty} \leq B$ to an untreated unit in the same cube without ever running out of untreated units, and we will still have leftover untreated units to match the treated units with $\| X_i \|_{\infty} > B$.  Let $\bar{m}_{\epsilon}$ be any matching that does this, and let $\bar{m}_{\epsilon}$ be defined arbitrarily when any of the above events fails.

When (\ref{eq:good_event1}) and (\ref{eq:good_event2}) both occur, the average matching discrepancy is guaranteed to be less than $\epsilon$.  This is because all treated observations in cubes $C_j$ find good matches and very few treated observations fall outside these cubes:
\begin{align*}
\frac{1}{N_1} \sum_{Z_i = 1} \| X_i - X_{\bar{m}_{\epsilon}(i)} \|_2 &= \frac{1}{N_1} \sum_{j = 1}^m \sum_{\substack{i \in C_j\\Z_i = 1}} \| X_i - X_{\bar{m}_{\epsilon}(i)} \|_2 + \frac{1}{N_1} \sum_{Z_i = 1} \| X_i - X_{\bar{m}_{\epsilon}(i)} \|_2 \mathbf{1} \{ \| X_i \|_{\infty} > B \}\\
&\leq \frac{1}{N_1} \sum_{j = 1}^m \sum_{\substack{i \in C_j\\Z_i=1}} \sqrt{d} \epsilon / 4 \sqrt{d} + \left\{ \frac{1}{N_1} \sum_{Z_i = 1} ( \| X_i \|_2 + \| X_{\bar{m}_{\epsilon}(i)} \|_2 )^2 \right\}^{1/2} \left( \frac{1}{N_1} \sum_{Z_i = 1} \mathbf{1} \{ \|X_i\|_{\infty} > B \} \right)^{1/2}\\
&\leq \frac{1}{N_1} \sum_{Z_i = 1} \epsilon/4 + \frac{n}{N_1} \left\{ \frac{2}{n} \sum_{i = 1}^n \| X_i \|^2 \right\}^{1/2} \left( \frac{1}{n} \sum_{i = 1}^n \mathbf{1} \{ \|X_i \|_{\infty} > B \} \right)^{1/2}\\
&\leq \epsilon/4 + \frac{2 \sqrt{8}}{p} \E[\|X\|^2]^{1/2} \{ P(X \not \in [-B, B]^d) \}^{1/2}\\
&\leq \epsilon/4 + \epsilon/4\\
&= \epsilon/2.
\end{align*}

It remains to show that the events (\ref{eq:good_event1}) and (\ref{eq:good_event2}) both occur with probability tending to one.  For (\ref{eq:good_event1}), this follows immediately from four applications of the law of large numbers.  Meanwhile, the following calculation shows that the ratio of untreated to treated observations in any cube $C_j$ with $P(X \in C_j) > 0$ converges to number strictly larger than one:
\begin{align*}
\frac{\sum_{i = 1}^n \mathbf{1} \{ X_i \in C_j, Z_i = 0 \}}{\sum_{i = 1}^n \mathbf{1} \{ X_i \in C_j, Z_i = 1 \}} &\xrightarrow{P} \frac{P(X \in C_j, Z = 0)}{P(X \in C_j, Z = 1)} \\
&= 1 + \frac{P(X \in C_j, Z = 0) - P(X \in C_j, Z = 1)}{P(X \in C_j, Z = 1)}\\
&= 1 + \frac{1}{P(X \in C_j, Z = 1)} \int_{C_j} \{ P(Z = 0 \mid X = x) - P(Z = 1 \mid X = x) \} \, \d P_X(x) \\
&= 1 + \frac{1}{P(X \in C_j, Z = 1)} \int_{C_j} \{ 1 - 2 e(x) \} \, \d P_X(x)\\
&> 1. 
\end{align*}
Here, the last step follows from the assumption that $P \{ e(X) < 0.5 \} = 1$.  Thus, we have shown that $C_j$ contains at least as many untreated units as treated units with probability tending to one.  Since there are only finitely many cubes, the probability that (\ref{eq:good_event2}) tends to one as well.  Thus, we have shown $\P( \sum_{Z_i = 1} \| X_i - X_{\bar{m}_{\epsilon}(i)} \| / N_1 > \epsilon) \rightarrow 0$. 

Now, we are ready to prove the lemma.  Let $\epsilon > 0$ be arbitrary.  By \Cref{lemma:mahalanobis_euclidean_equivalence}, $\frac{1}{N_1} \sum_{Z_i = 1} d_M(X_i, X_{\bar{m}_{\epsilon}(i)}) \leq \kappa(\mathbf{\Sigma}) \epsilon$, with probability approaching one, where $d_M(x, x')$ is the estimated Mahalanobis distance and $\kappa( \mathbf{\Sigma})$ is the constant introduced in \Cref{lemma:mahalanobis_euclidean_equivalence}.  Since $\bar{m}_{\epsilon}$ is a feasible pair matching scheme and $m$ is the optimal pair matching scheme, we must have $\tfrac{1}{N_1} \sum_{Z_i = 1} d_M(X_i, X_{m(i)}) \leq \kappa(\mathbf{\Sigma}) \epsilon$ on the same event.  Applying \Cref{lemma:mahalanobis_euclidean_equivalence} again gives $\sum_{Z_i = 1} \| X_i - X_{m(i)} \| / N_1 \leq \kappa(\mathbf{\Sigma})^2 \epsilon$ with probability approaching one.  Since $\epsilon > 0$ is arbitrary, this proves the result.
\end{proof}

\begin{corollary}
\label{corollary:paired_discrepancy_Lp}
If $P \{ e(X) < 0.5 \} = 1$, then $\sum_{Z_i = 1} | g(X_i) - g(X_{m(i)}) |^2 / N_1 = o_P(1)$ whenever $\E \{ |g(X)|^2 \} < \infty$.
\end{corollary}

\begin{proof}
Let $g$ satisfy $\E\{ |g(X)|^2 \} < \infty$ and let $\epsilon > 0$ be arbitrary.  Since Lipschitz functions are dense in $L^2$, there exists a function $h$ satisfying $\E \{ | g(X) - h(X)|^2 \} < \epsilon$ and $|h(x) - h(x')| \leq L \| x - x' \|$ for all $x, x'$ and some $L < \infty$.  Thus, we may write:
\begin{align*}
\frac{1}{N_1} \sum_{Z_i = 1} | g(X_i) - g(X_{m(i)})|^2 &\leq \frac{4}{N_1} \sum_{Z_i = 1} | g(X_i) - h(X_i)|^2 + \frac{4}{N_1} \sum_{Z_i = 1} | g(X_{m(i)}) - h(X_{m(i)})|^2\\
&+ \frac{4}{N_1} \sum_{Z_i = 1} | h(X_i) - h(X_{m(i)})|^2\\
&\leq \underbrace{\frac{4n}{N_1} \frac{1}{n} \sum_{i = 1}^n | g(X_i) - h(X_i)|^2}_a + \underbrace{4 L^2 \frac{1}{N_1} \sum_{i = 1}^n \| X_i - X_{m(i)} \|^2}_b.
\end{align*}
By the law of large numbers, the term $a$ converges to $(4/p) \E \{ |g(X) - h(X)|^2 \} < 4 \epsilon / p$.  Meanwhile, the following calculation shows that term $b$ is vanishing:
\begin{align*}
\frac{1}{N_1} \sum_{Z_i = 1} \| X_i - X_{m(i)} \|^2 &= \frac{1}{N_1} \sum_{Z_i = 1} \| X_i - X_{m(i)} \|^{0.5} \| X_i - X_{m(i)} \|^{1.5}\\
&\leq \left( \frac{1}{N_1} \sum_{Z_i = 1} \| X_i - X_{m(i)} \| \right)^{1/2} \left( \frac{1}{N_1} \sum_{Z_i = 1} \| X_i - X_{m(i)} \|^3 \right)^{1/2}\\
&\leq \left( \frac{1}{N_1} \sum_{Z_i = 1} \| X_i - X_{m(i)} \| \right)^{1/2} \left( \frac{8}{N_1} \sum_{Z_i = 1} \| X_i \|^3 + \| X_{m(i)} \|^3 \right)^{1/2}\\
&\leq \left( \frac{1}{N_1} \sum_{Z_i = 1} \| X_i - X_{m(i)} \| \right)^{1/2} \left( \frac{8}{N_1} \sum_{i = 1}^n \| X_i \|^3 \right)^{1/3}
\end{align*}
By \Cref{proposition:mahalanobis_discrepancy_vanishes} and our moment assumptions on $\| X \|$, this upper bound tends to zero.  In particular, it less than $\epsilon$ with probability approaching one.

Hence, combining $a$ and $b$ gives $\sum_{Z_i = 1} | g(X_i) - g(X_{m(i)})|^2 / N_1 \leq 4 \epsilon / p + \epsilon$ with probability approaching one.  Since $\epsilon$ is arbitrary, this proves the lemma.
\end{proof}

\subsection{Theoretical background on matching with replacement}

\begin{lemma}
\label{lemma:Kn_moments}
Let $K_{i,n}$ be the number of treated observations matched to the $i$-th observation under the matching scheme (\ref{eq:matching_with_replacement}).  Then for every $q \geq 0$, $\sup_{n\geq 1}\E( K_{i,n}^q )$ is bounded by a constant $\kappa_q$ depending only on $q, \delta$ and $d$.
\end{lemma}

\begin{proof}
This result is very similar to \cite[Lemma 3.(iii)]{abadie_imbens2006}, so we defer the proof to Appendix \ref{section:additional_proofs}.  The proof is quite different from theirs, as we do not make as many regularity assumptions.
\end{proof}

\begin{lemma}
\label{lemma:matched_moments}
For any exponent $q > 1$ and any function $g : \R^d \rightarrow \R$ with $\E\{ | g(X)|^q \} < \infty$, we have:
\begin{align*}
\E[ | g(X_{m_r(1)})| \mathbf{1} \{ N_0 > 0 \}] \leq C(d, \delta, q) \E\{ | g(X)|^q \}^{1/q}
\end{align*}
for some constant $C(d, \delta, q) < \infty$ depending only on the dimension $d$, the overlap parameter $\delta$, and the exponent $q$.
\end{lemma}

\begin{proof}
The proof is based on \cite[Lemma 6.3]{gyorfi_etal2002}, and exploits the exchangeability of the observations:
\begin{align*}
\E [ | g(X_{m_r(1)})| \mathbf{1} \{ N_0 > 0 \}] &= \frac{1}{n} \sum_{i = 1}^n \E [ | g(X_{m_r(j)}) \mathbf{1} \{ N_0 > 0 \}]\\
&\leq \frac{1}{n} \sum_{j = 1}^n \sum_{i = 1}^n \E [ | g(X_i)| \mathbf{1} \{ m_r(j) = i \}]\\
&= \frac{1}{n} \sum_{i = 1}^n \E \left[ | g(X_i) | \sum_{j = 1}^n \mathbf{1} \{ m_r(j) = i \} \right] \\
&= \frac{1}{n} \sum_{i = 1}^n \E\{ | g(X_i)| K_{i,n} \}\\
&= \E \{ | g(X_1)| K_{1,n} \}.
\end{align*}
Now the conclusion follows from H\"older's inequality and \Cref{lemma:Kn_moments}.
\end{proof}

\begin{lemma}
\label{lemma:match_quality}
$\| X_1 - X_{m_r(1)} \| \mathbf{1} \{ N_0 > 0 \} \xrightarrow{P} 0$.
\end{lemma}

\begin{proof}
First, we prove that $\| X_1 - X_{m_e(1)} \| \mathbf{1} \{ N_0 > 0 \} \xrightarrow{P} 0$ where $m_e(\cdot)$ forms matches using Euclidean distance rather than Mahalanobis distance.  We do this by mimicking the proof of \cite[Lemma 6.1]{gyorfi_etal2002}.  For any $\epsilon > 0$, we have:
\begin{align*}
\P( \| X_1 - X_{m_e(1)} \| \mathbf{1} \{ N_0 > 0 \} > \epsilon) &\leq \P( \| X_1 - X_{m_e(1)} \| \mathbf{1} \{ N_0 > 0 \} > \epsilon \mid Z_1 = 1)\\
&= \int_{\R^d} \P( \| x - X_{m_e(1)} \| \mathbf{1} \{ N_0 > 0 \} > \epsilon \mid Z_1 = 1, X_1 = x) \, \d P_X(x \mid Z = 1)\\
&\leq \int_{\R^d} \prod_{i = 2}^n P \{ Z = 1 \text{ or } X \not \in \bar{\mathbb{B}}_{\epsilon}(x) \} \, \d P_X(x \mid Z = 1)\\
&= \int_{\R^d} P \{ Z = 1 \text{ or } X \not \in \bar{\mathbb{B}}_{\epsilon}(x) \}^{n -1 } \, \d P_X(x \mid Z = 1).
\end{align*}
For each $x$ in the support of $P(X \mid Z = 1)$, the probability $P \{ Z = 1 \text{ or } X \not \in \mathbb{B}_{\epsilon}(x) \}$ is strictly less than one by overlap and the definition of support.  Therefore, $P\{ Z =1 \text{ or }X \not \in \mathbb{B}_{\epsilon}(x) \}^{n - 1} \rightarrow 0$ for $P( X \mid Z = 1)$-almost every $x$.  Hence, by the dominated convergence theorem, $\P( \| X_1 - X_{m_e(1)} \| > \epsilon) \rightarrow 0$.  Since $\epsilon$ is arbitrary, this means $\| X_1 - X_{m_e(1)} \| = o_P(1)$ under Euclidean matching.  

To extend this result to Mahalanobis matching, we reason as follows.  Suppose that the event in \Cref{lemma:mahalanobis_euclidean_equivalence} occurs, so that Mahalanobis and Euclidean distances are equivalent.  Then we may write:
\begin{align*}
\| X_1 - X_{m_r(1)} \| &\leq \kappa(\mathbf{\Sigma}) d_M(X_1, X_{m_r(1)}) \leq \kappa(\mathbf{\Sigma}) d_M(X_1, X_{m_e(1)}) \leq \kappa(\mathbf{\Sigma})^2 \| X_1 - X_{m_e(1)} \|.
\end{align*}
Since this bound holds with probability approaching one and $\| X_1 - X_{m_e(1)} \| \xrightarrow{P} 0$, we have $\| X_1 - X_{m_r(1)} \| \xrightarrow{P} 0$ as well.
\end{proof}

\begin{lemma}
\label{lemma:replaced_Lp_convergence}
Let $g : \R^d \rightarrow \R$ satisfy $\E \{ | g(X)|^q \} < \infty$ for some $q > 1$.  Then $\E[| g(X_1) - g(X_{m_r(1)})|^r \mathbf{1} \{ N_0 > 0 \}] \rightarrow 0$ for any $r \in [1, q)$.
\end{lemma}

\begin{proof}
Let $\epsilon > 0$ be arbitrary.  There exists a bounded continuous function $h$ such that $\E \{ | h(X) - g(X)|^q \} < \epsilon$, since $C_b(\R^d)$ is dense in $L^q$ for any $q$.  Thus, we may write:
\begin{align*}
\E[ | g(X_1) - g(X_{m_r(1)})|^r \mathbf{1} \{ N_0 > 0 \}] &\leq 3^{r - 1}\E[| g(X_1) - h(X_1)|^r \mathbf{1} \{ N_0 > 0 \}]\\
&+ 3^{r-1} \E[| g(X_{m_r(1)}) - h(X_{m_r(1)})|^r \mathbf{1} \{ N_0 > 0 \}]\\
&+ 3^{r - 1} \E[ | h(X_1) - h(X_{m_r(1)})|^r \mathbf{1} \{ N_0 > 0 \}]\\
&\leq 3^{r - 1} \epsilon^{r/q} + 3^{r - 1} C(d, \delta, q/r) \E [ | g(X) - h(X)|^q] &\text{\Cref{lemma:matched_moments}}\\
&+ 3^{r - 1} \E[ |h(X_1) - h(X_{m_r(1)})|^r \mathbf{1} \{ N_0 > 0 \}]\\
&\leq 3^{r - 1} \epsilon^{r/q} + 3^{r - 1} C(d, \delta, q/r) \epsilon\\
&+ 3^{r - 1} \E[| h(X_1) - h(X_{m_r(1)}|^r \mathbf{1} \{ N_0 > 0 \}].
\end{align*}
Since $h$ is continuous and $X_{m_r(1)} \xrightarrow{P} X_1$ (\Cref{lemma:match_quality}), $|h(X_{m_r(1)}) - h(X_1)|^r \mathbf{1} \{ N_0 > 0 \} = o_P(1)$.  Moreover, this random variable has a uniformly bounded higher moment by \Cref{lemma:matched_moments}.  Thus, Vitali's convergence theorem gives $\E[| h(X_1) - h(X_{m_r(1)})|^r \mathbf{1} \{ N_0 > 0 \}] \rightarrow 0$.  In particular, for all large $n$ we have:
\begin{align*}
\E[| g(X_1) - g(X_{m_r(1)})|^r \mathbf{1} \{ N_0 > 0 \}] &\leq 3^{r - 1} \epsilon^{r/q} + 3^{r - 1} C(d, \delta, q/r) \epsilon + \epsilon.
\end{align*}
Since $\epsilon$ is arbitrary, this proves the lemma.
\end{proof}

\begin{corollary}
\label{corollary:empirical_Lp_convergence}
Let $g : \R^d \rightarrow \R$ satisfy $\E \{ |g(X)|^q \} < \infty$ for some $q > 1$.  Then $\sum_{Z_i = 1} | g(X_i) - g(X_{m_r(i)})|^r / N_1 = o_P(1)$ for any $r \in [1, q)$.
\end{corollary}

\begin{proof}
This is immediate from \Cref{lemma:replaced_Lp_convergence} and Markov's inequality.
\end{proof}

\begin{lemma} \label{lemma:replaced_wlln}
For any $q > 1$ and any function $g : \R^d \rightarrow \R$ with $\E \{ | g(X)|^q \} < \infty$, the following hold:
\begin{align*}
\frac{1}{N_1} \sum_{Z_i = 1} g(X_i) \xrightarrow{P} \E\{ g(X) \mid Z = 1 \} \quad \text{and} \quad \frac{1}{N_1} \sum_{Z_i = 1} g(X_{m_r(i)}) \xrightarrow{P} \E \{ g(X) \mid Z = 1 \}.
\end{align*}
\end{lemma}

\begin{proof}
The convergence $\sum_{Z_i = 1} g(X_i) / N_1 \xrightarrow{P} \E \{ g(X) \mid Z = 1 \}$ is simply the law of large numbers and the continuous mapping theorem.  For the other average, we reason as follows:
\begin{align*}
\left| \frac{1}{N_1} \sum_{Z_i = 1} g(X_{m_r(i)}) - \E \{ g(X) \mid Z = 1\} \right| &\leq  \frac{1}{N_1} \sum_{Z_i = 1} | g(X_i) - g(X_{m_r(i)})| + \left| \frac{1}{N_1} \sum_{Z_i = 1} g(X_i) - \E \{ g(X) \mid Z = 1\} \right| 
\end{align*}
The two terms in the upper bound tend to zero by \Cref{corollary:empirical_Lp_convergence} and the law of large numbers, respectively.
\end{proof}

\subsection{Randomization test critical values}

Throughout this section, we will make use of the following quantities:
\begin{align}
s_n^2 = \frac{1}{N_1} \sum_{i \in \M} \sigma^2(X_i), \quad \hat{s}_n^2 = \frac{1}{N_1} \sum_{Z_i = 1} ( \epsilon_i - \epsilon_{m(i)})^2, \quad s^2 = 2 \E \{ \sigma^2(X) \mid Z = 1 \} \label{eq:randomization_variances}
\end{align}

\subsubsection{The difference-of-means statistic}

\begin{lemma}
\label{lemma:sigma_star_limit}
Let $P \in H_0$ satisfy $P \{ e(X) < 0.5 \} = 1$.  Then $s_n$ and $\hat{s}_n$ both converge to $s$ in probability.
\end{lemma}

\begin{proof}
The following simple calculation shows that $s_n$ converges to the claimed limit.
\begin{align*}
s_n^2 &= \frac{1}{N_1} \sum_{Z_i = 1} \{ \sigma^2(X_i) + \sigma^2(X_{m(i)}) \}\\
&= \frac{1}{N_1} \sum_{Z_i = 1} 2 \sigma^2(X_i) + 
\underbrace{\frac{1}{N_1} \sum_{Z_i = 1} \{ \sigma^2(X_{m(i)}) - \sigma^2(X_i) \}}_{= o_P(1) \text{ by \Cref{corollary:paired_discrepancy_Lp}}}\\
&\xrightarrow{P} 2 \E \{ \sigma^2(X) \mid Z = 1 \}\\
&= s^2.
\end{align*}

To prove that $\hat{s}_n$ converges to the same limit, we will show that the difference between $\hat{s}_n^2$ and $s_n^2$ is vanishing.  Let $\F_n = \sigma(\{ (X_i, Z_i) \}_{i \leq n})$ and observe that $\E( \hat{s}_n^2 - s_n^2 \mid \F_n) = 0$.  Moreover, the conditional variance can be bounded as follows:
\begin{align*}
\Var( \hat{s}_n^2 - s_n^2 \mid \F_n) &= \Var \left( \frac{1}{N_1} \sum_{Z_i = 1} ( \epsilon_i - \epsilon_{m(i)})^2 \, \bigg| \, \F_n \right)\\
&\leq \frac{1}{N_1^2} \sum_{Z_i = 1} \E\{  ( \epsilon_i - \epsilon_{m(i)})^4 \mid \F_n \}\\
&\leq \frac{8}{N_1^2} \sum_{Z_i = 1} \{ \E( | \epsilon_i |^4 \mid \F_n) + \E( | \epsilon_{m(i)}|^4 \mid \F_n) \}\\
&\leq \frac{8}{N_1} \frac{n}{N_1} \frac{1}{n} \sum_{i = 1}^n \E( | \epsilon_i|^4 \mid \F_n)
\end{align*}
Since $\sum_{i = 1}^n \E[ \epsilon_i^4 \mid \F_n] / n = O_P( \E[ \epsilon^4]) = O_P(1)$ by Markov's inequality and our assumption that $Y$ has four moments, this upper bound tends to zero at rate $O_P(1/N_1) = O_P(1/n)$.  Hence, by \Cref{lemma:conditional_wlln}, $\hat{s}_n^2 - s_n^2 = o_P(1)$.
\end{proof}

\begin{lemma}
\label{lemma:dm_critical_value}
Assume $P \in H_0$.  For any $\alpha \in (0, 1)$, define the randomization quantile
\begin{align}
\hat{q}_{1 - \alpha}^{\dm} = \inf \{ t \in \R \, : \, \P( | \hat{\tau}_*^{\dm}| \leq t \mid \mathcal{D}_n) \geq 1 - \alpha \} \label{dm_quantile}
\end{align}
Then $\hat{q}_{1 - \alpha}^{\dm} = O_P(n^{-1/2})$.  If $P$ also satisfies $P \{ e(X) < 0.5 \} = 1$, then we also have $\hat{q}_{1 - \alpha}^{\dm} = z_{1 - \alpha/2} s N_1^{-1/2} + o_P(n^{-1/2})$, where $s$ is defined in \Cref{eq:randomization_variances}.
\end{lemma}

\begin{proof}
First, we prove that $\hat{q}_{1 - \alpha}^{\dm} = O_P(n^{-1/2})$.  Since $\E( \hat{\tau}_*^{\dm} \mid \mathcal{D}_n) = 0$, Chebyshev's inequality applied conditionally gives $\hat{q}_{1 - \alpha}^{\dm} \leq \{ \Var( \hat{\tau}_*^{\dm} \mid \mathcal{D}_n) / \alpha \}^{1/2}$.  This may be further bounded as follows:
\begin{align*}
\{ \Var( \hat{\tau}_*^{\dm} \mid \mathcal{D}_n) / \alpha \}^{1/2} &= \left( \frac{1}{\alpha} \frac{1}{N_1^2} \sum_{i \in \M} Y_i^2 \right)^{1/2} \leq \frac{1}{\sqrt{n \alpha}} \frac{n}{N_1} \left( \frac{1}{n} \sum_{i = 1}^n Y_i^2 \right)^{1/2}
\end{align*}
Since $n / N_1 \xrightarrow{P} 1/p$ and $\tfrac{1}{n} \sum_{i = 1}^n Y_i^2 \xrightarrow{P} \E(Y^2) < \infty$, this upper bound is of order $O_P(n^{-1/2})$.

Next, we give a more precise result in the case $P \{ e(X) < 0.5 \} = 1$.  Consider the following stochastic expansion of $\sqrt{N_1} \hat{\tau}_* / s$:
\begin{align*}
\sqrt{N_1} \hat{\tau}_* / s &= \sqrt{N_1} \hat{\tau}_* / \hat{s}_n + o_P(1)\\
&= \frac{1}{\hat{s}_n} \frac{1}{\sqrt{N_1}} \sum_{Z_i = 1} (2 Z_i^* - 1) (Y_i - Y_{m(i)}) + o_P(1)\\
&= \underbrace{\frac{1}{
\hat{s}_n} \frac{1}{\sqrt{N_1}} \sum_{Z_i = 1} (2 Z_i^* - 1) \{ \mu(X_i) - \mu(X_{m(i)}) \}}_a + \underbrace{\frac{1}{
\hat{s}_n} \frac{1}{\sqrt{N_1}} \sum_{Z_i = 1} (2 Z_i^* - 1) ( \epsilon_i - \epsilon_{m(i)})}_b
\end{align*}

Term $a$ is mean zero conditional on $\mathcal{D}_n$ and has conditional variance $\sum_{Z_i = 1} | \mu(X_i) - \mu(X_{m(i)})|^2 / N_1 \hat{s}_n^2$.  By \Cref{corollary:paired_discrepancy_Lp}, $\sum_{Z_i = 1} | \mu(X_i) - \mu(X_{m(i)} |^2 / N_1 = o_P(1)$ and by \Cref{lemma:sigma_star_limit}, $1/\hat{s}_n^2 = O_P(1)$.  Hence, this conditional variance tends to zero in probability, so $a = o_P(1)$ by \Cref{lemma:conditional_wlln}.

Meanwhile, we will show that term $b$ converges weakly in probability to the standard normal distribution.  Observe that $\hat{s}_n^2$ is precisely the variance of $\sum_{Z_i = 1} (2 Z_i^* - 1) ( \epsilon_i - \epsilon_{m(i)}) / \sqrt{N_1}$, so we may apply the Berry-Esseen theorem (\Cref{theorem:berry_esseen}) conditionally on $\mathcal{D}_n$ to $b$ and obtain:
\begin{align*}
\sup_{t \in \R} |\P (b \leq t \mid \mathcal{D}_n) - \Phi(t)| &\leq \frac{C}{\sqrt{N_1}} \frac{1}{\sigma_*^3} \frac{1}{N_1} \sum_{Z_i = 1} \E( | \epsilon_i - \epsilon_{m(i)} |^3 \mid \mathcal{D}_n) \leq \frac{4 C}{\sqrt{N_1}} \frac{1}{\hat{s}_n^3} \frac{n}{N_1} \frac{1}{n} \sum_{i = 1}^n \E( | \epsilon_i|^3 \mid \mathcal{D}_n).
\end{align*}
Since $\sum_{i = 1}^n \E( | \epsilon_i|^3 \mid \mathcal{D}_n)/n = O_P \{ \E(| \epsilon^3|) \} = O_P(1)$ by Markov's inequality and $1/\hat{s}_n^3 = O_P(1)$ by \Cref{lemma:sigma_star_limit}, this upper bound tends to zero in probability.  Hence, $b \rightsquigarrow_* N(0, 1)$.

Since $\sqrt{N_1} \hat{\tau}_* / s = b + o_P(1)$ and $b \rightsquigarrow_* H$ for $H \sim N(0, 1)$, \Cref{lemma:randomization_slutsky} implies $\sqrt{N_1} \hat{\tau}_* / s \rightsquigarrow_* H$ as well.  By the continuous mapping theorem for weak convergence in probability, $| \sqrt{N_1} \hat{\tau}_* / s | \rightsquigarrow |H|$.  Weak convergence in probability to a continuous limit distribution implies convergence of quantiles \cite[Lemma 11.2.1.(ii)]{tsh}, so we conclude that:
\begin{align*}
\inf \{ t \in \R \, : \, \P( | \sqrt{N_1} \hat{\tau}_* / s | \leq t \mid \mathcal{D}_n) \geq 1 - \alpha \} \xrightarrow{P} \inf \{ t \in \R \, : \, \P_{H \sim N(0, 1)}(|H| \leq t) \geq 1 - \alpha \} = z_{1 - \alpha/2}.
\end{align*}
Thus, the basic calculus of quantiles allows us to write:
\begin{align*}
\hat{q}_{1 - \alpha}^{\dm} &= \inf \{ t \in \R \, : \, \P(| \hat{\tau}_* | \leq t \mid \mathcal{D}_n) \geq 1 - \alpha \}\\
&= s N_1^{-1/2} \inf \{ t \in \R \, : \, \ \P( | \sqrt{N_1} \hat{\tau}_* / s | \leq t \mid \mathcal{D}_n) \geq 1 - \alpha \}\\
&= s N_1^{-1/2} \{ z_{1 - \alpha/2} + o_P(1) \}\\
&= z_{1 - \alpha/2} s N_1^{-1/2} + o_P(n^{-1/2}).
\end{align*}
\end{proof}

\subsubsection{The regression-adjusted statistic}

\begin{lemma}
\label{lemma:B_invertibility}
Let $\phi : \R^d \rightarrow \R^{\ell}$ satisfy $\E_P \{ \| \phi(X) \|^q \} < \infty$ for some $q > 4$ and also $\Var \{ \phi(X) \mid Z = 1 \} \succ \mathbf{0}$.  Define $\psi_{\phi}(x, z) = (z - 0.5, 1, \phi(x))$.  Let $\mathbf{B} = \sum_{i \in \M} \psi_{\phi}(X_i, Z_i) \psi_{\phi}(X_i, Z_i)^{\top} / 2 N_1$ and $\mathbf{B}_* = \sum_{i \in \M} \psi_{\phi}(X_i, Z_i^*) \psi_{\phi}(X_i, Z_i^*)^{\top} / 2 N_1$.  Then $\mathbf{B}$ and $\mathbf{B}_*$ are invertible with probability tending to one.
\end{lemma}

\begin{proof}
First, we prove the result for $\mathbf{B}_*$.  Repeated applications of \Cref{lemma:conditional_wlln} and \Cref{corollary:empirical_Lp_convergence} allow us to derive the following convergence:
\begin{align*}
\mathbf{B}_* &\succeq \frac{1}{2 N_1} \sum_{Z_i = 1} \psi_{\phi}(Z_i^*, X_i) \psi_{\phi}(Z_i^*, X_i)^{\top} \xrightarrow{P} 
\frac{1}{2} \left[ 
\begin{array}{cc}
0.25 &\mathbf{0}^{\top}\\
\mathbf{0} &\E \{ (1, \phi(X))(1, \phi(X))^{\top} \mid Z = 1 \}
\end{array}
\right] 
\end{align*}
Under the assumption that $\Var \{ \phi(X) \mid Z = 1 \} \succ \mathbf{0}$, the matrix on the far right-hand side of the preceding display is invertible.  Therefore, $\mathbf{B}_*$ is invertible with probability approaching one.

Next, we prove the result for $\mathbf{B}$, which requires \Cref{lemma:replaced_design_matrix} below.  Notice that any observation $i$ matched by the matching-with-replacement scheme (\ref{eq:matching_with_replacement}) will also be matched by the optimal pair-matching scheme (\ref{total_mahalanobis}).  Thus, if $m_r(i)$ denotes the nearest untreated neighbor of observation $i$, we have:
\begin{align*}
\mathbf{B} &\succeq \frac{1}{2 N_1} \sum_{Z_i = 1} \psi_{\phi}(X_i, Z_i) \psi_{\phi}(X_i, Z_i)^{\top} + \frac{1}{2 N_1} \sum_{j = m_r(i) \text{ for some treated $i$}} \psi_{\phi}(X_i, Z_i) \psi_{\phi}(X_i, Z_i)^{\top}\\
&\succeq \frac{1}{\max_i K_{i,n}} \underbrace{\frac{1}{2 N_1} \sum_{Z_i = 1} \{ \psi_{\phi}(X_i, Z_i) \psi_{\phi}(X_i, Z_i)^{\top} + \psi_{\phi}(X_{m_r(i)}, Z_{m_r(i)}) \psi_{\phi}(X_{m_r(i)}, Z_{m_r(i)})^{\top} \}}_{:= \mathbf{B}_r}
\end{align*}
Here, $K_{i,n} = \sum_{i = 1}^n Z_j \mathbf{1} \{ j = m_r(i) \}$ is the number of untreated observations matched to observation $i$.  \Cref{lemma:replaced_design_matrix} below shows that $\mathbf{B}_r$ converges to an invertible matrix.  Since $\max_i K_{i,n}$ is finite, $\mathbf{B} \succeq \mathbf{B}_r / \max_i K_{i,n} \succ \mathbf{0}$ with probability tending to one.
\end{proof}

\begin{lemma}
\label{lemma:paired_design_matrix_properties}
Let $\mathbf{B}$ and $\mathbf{B}_*$ be as in \Cref{lemma:B_invertibility}.  Then $e_1^{\top} \mathbf{B}_*^{-1} \xrightarrow{P} (4, \mathbf{0})$.  If $P \{ e(X) < 0.5 \} = 1$, then we also have:
\begin{align*}
\mathbf{B} \xrightarrow{P} \left[ 
\begin{array}{cc}
0.25 &\mathbf{0}^{\top}\\
\mathbf{0} &\E\{ (1, \phi(X))(1, \phi(X))^{\top} \mid Z = 1 \}
\end{array}
\right] 
\end{align*}
\end{lemma}

\begin{proof}
First, we show that $e_1^{\top} \mathbf{B}_*^{-1} \xrightarrow{P} (4, \mathbf{0})$.  By \Cref{lemma:B_invertibility}, $e_1^{\top} \mathbf{B}^{-1}$ exists with probability tending to one.  When it exist, the partitioned matrix inversion formula gives the following explicit expression:
\begin{align}
e_1^{\top} \mathbf{B}_*^{-1} &= \frac{1}{0.25 - \Delta_*^{\top} \hat{\mathbf{\Sigma}}_{\phi}^{-1} \Delta_*} ( 1, -(0, \Delta_*^{\top} \hat{\mathbf{\Sigma}}_{1, \phi}^{-1} )). \label{eq:e1topBstarinverse}
\end{align}
where $\Delta_* = \sum_{i \in \M} (Z_i^* - 0.5) \phi(X_i) / 2N_1$, $\hat{\mathbf{\Sigma}}_{\phi} = \sum_{i \in \M} \phi(X_i) \phi(X_i)^{\top} / 2 N_1$, and $\hat{\mathbf{\Sigma}}_{1, \phi} = \sum_{i \in \M} ( 1, \phi(X_i) ) ( 1, \phi(X_i) ) / 2 N_1$.

A straightforward application of \Cref{lemma:conditional_wlln} shows that $\Delta_* = o_P(1)$.  Moreover, the following calculation shows that $\hat{\mathbf{\Sigma}}_{\phi}^{-1}$ and $\hat{\mathbf{\Sigma}}_{1,\phi}^{-1}$ are stochastically bounded:
 \begin{align*}
    \| \hat{\mathbf{\Sigma}}_{\phi}^{-1} \|_{\textup{op}} & \leq \lambda_{\min} \left\{ \frac{1}{2 N_1} \sum_{Z_i = 1} \phi(X_i) \phi(X_i)^{\top} \right\}^{-1} \xrightarrow{P} \lambda_{\min} [ 0.5 \E \{ \phi(X) \phi(X)^{\top} \mid Z = 1 \}]^{-1} < \infty\\
    \| \hat{\mathbf{\Sigma}}_{1,\phi}^{-1} \|_{\textup{op}} & \leq \lambda_{\min} \left[ \frac{1}{2 N_1} \sum_{Z_i = 1} \{ 1, \phi(X_i) \} \{ 1, \phi(X_i) \}^{\top} \right]^{-1} \xrightarrow{P} \lambda_{\min}( 0.5 \E[ \{ 1, \phi(X) \} \{ 1, \phi(X) \}^{\top} \mid Z = 1])^{-1} < \infty.
\end{align*}
Thus, $\Delta_*^{\top} \hat{\mathbf{\Sigma}}_{\phi}^{-1} \Delta_* = o_P(1)$ and $(0, \Delta_*)^{\top} \hat{\mathbf{\Sigma}}_{1,\phi}^{-1} = o_P(1)$.  Combining these with the formula (\ref{eq:e1topBstarinverse}) gives $e_1^{\top} \mathbf{B}_*^{-1} = (4, \mathbf{0}) + o_P(1)$, as claimed.  For the second claim, simply apply \Cref{corollary:paired_discrepancy_Lp} to each coordinate of $\mathbf{B}$.
\end{proof}

\begin{lemma}
\label{lemma:ancova_covariance_properties}
Assume $P \in H_0$. Let $g : \R^d \rightarrow \R^k$ be a bounded function and let $\{ h_n \} \subset \R^k$ satisfy $\| h_n \| \leq C n^{-1/2}$.  Let $\phi : \R^d \rightarrow \R^{\ell}$ be any function with $\E \{ \| \phi(X) \|^q \} < \infty$ for some $q > 4$, and set $\psi_{\phi}(x, z) = (z - 0.5, 1, \phi(x))$.  Let $\mathbf{r} = \sum_{i \in \M} \psi_{\phi}(X_i, Z_i) \{ \epsilon_i + h_n^{\top} g(X_i) \} / 2 \sqrt{N_1}$ and $\mathbf{r}_* = \sum_{i \in \M} \psi_{\phi}(X_i, Z_i^*) \{ \epsilon_i + h_n^{\top} g(X_i) \} / 2 \sqrt{N_1}$.  Then $\mathbf{r} = O_P(1)$ and $\mathbf{r}_* = O_P(1)$.
\end{lemma}

\begin{proof}
We will show that $\sum_{i \in \M} h(X_i, Z_i, Z_i^*) \{ \epsilon_i + h_n^{\top} g(X_i) \} /2 \sqrt{N_1} = O_P(1)$ whenever $\E \{ h(X_i, Z_i, Z_i^*)^4 \}\footnote{ To make sense of this expectation, we define $Z_i^*$ to have a $\textup{Bernoulli}(0.5)$ distribution independently of $(X_i, Y_i, Z_i)$ even for unmatched observations $i$.} < \infty$.  The conclusion of the Lemma follows by applying this result to each of the coordinates or $(Z_i - 0.5, 1, X_i^{\top})$ of $(Z_i^* - 0.5, 1, X_i^{\top})$.  To do this, we decompose the quantity of interest into the sum of two terms:
\begin{align*}
\frac{1}{\sqrt{2 N_1}} \sum_{i \in \M} h(X_i, Z_i, Z_i^*) \{ \epsilon_i + h_n^{\top} g(X_i) \} &= \underbrace{\frac{1}{2 \sqrt{N_1}} \sum_{i \in \M} h(X_i, Z_i, Z_i^*) \epsilon_i}_a + \underbrace{\frac{1}{2 \sqrt{N_1}} \sum_{i \in \M} h(X_i, Z_i, Z_i^*) h_n^{\top} g(X_i)}_b.
\end{align*}
Conditional on $\mathcal{F}_n = \sigma( \{ (X_i, Z_i) \}_{i \leq n}, \{ Z_i^* \}_{i \in \M})$, term $a$ has mean zero and conditional variance of order one:
\begin{align*}
\Var (a \mid \F_n) &= \frac{1}{4 N_1} \sum_{i \in \M} h(X_i, Z_i, Z_i^*)^2 \sigma^2(X_i) \leq \frac{1}{4 N_1} \sum_{i = 1}^n h(X_i, Z_i, Z_i^*)^2 \sigma^2(X_i) \xrightarrow{P} \frac{1}{4p} \E\{ h(X, Z, Z^*)^2 \sigma^2(X) \}
\end{align*}
Thus, by Markov's inequality, $\sum_{i \in \M} h(X_i, Z_i, Z_i^*) \epsilon_i / 2 \sqrt{N_1} = O_P(1)$.  Meanwhile, $b$ may be bounded as follows:
\begin{align*}
\| b \| &\leq \frac{\| h_n \|}{2 \sqrt{N_1}} \sum_{i = 1}^n | h(X_i, Z_i, Z_i^*)| \cdot \| g(X_i) \| \leq \underbrace{C \left( \frac{n}{N_1} \right)^{1/2}}_{= O_P(1)} \underbrace{\frac{1}{n} \sum_{i = 1}^n | h(X_i, Z_i, Z_i^*)| \cdot \| g(X_i) \|}_{= O_P(1) \text{ by Markov}} = O_P(1).
\end{align*}
Thus, $a + b = O_P(1)$.
\end{proof}

\begin{lemma}
\label{lemma:paired_ancova_critical_value}
Let $P \in H_0$ satisfy the assumptions of \Cref{proposition:ols_sufficient}.  Then for any $\alpha \in (0, 1)$, we have $\sqrt{N_1} \hat{q}_{1 - \alpha,n}^{\reg} = s z_{1 - \alpha/2} + o_P(1)$ where $s^2 = 2 \E \{ \sigma^2(X) \mid Z = 1 \}$.
\end{lemma}

\begin{proof}
Let $\epsilon_i = Y_i - \mu(X_i)$, and let $\psi(x, z) = (z - 0.5, 1, x)$.  Set $\mathbf{B}_* = \sum_{i \in \M} \psi(X_i, Z_i^*) \psi(X_i, Z_i^*)^{\top} / 2 N_1$ and $\mathbf{r}_* = \sum_{i \in \M} \epsilon_i \psi(X_i, Z_i^*) / \sqrt{N_1}$.  By \Cref{lemma:B_invertibility} (applied with $\phi(x) = x$), $\mathbf{B}_*$ is invertible with probability approaching one.  On this event, standard least-squares theory gives us the formula $\hat{\tau}_*^{\reg} = e_1^{\top} \mathbf{B}_*^{-1} \mathbf{r}_*$.  \Cref{lemma:paired_design_matrix_properties} (applied with $\phi(x) = x$) shows that $e_1^{\top} \mathbf{B}_*^{-1} = (4, \mathbf{0}) + o_P(1)$ and \Cref{lemma:ancova_covariance_properties} shows that $\mathbf{r}_* = O_P(1)$.  Thus, we may write:
\begin{align*}
\sqrt{N_1} \hat{\tau}_*^{\reg} &= e_1^{\top} \mathbf{B}_*^{-1} \mathbf{r}_* = (4, \mathbf{0})^{\top} \mathbf{r} + o_P(1) = \frac{1}{\sqrt{N_1}} \sum_{i \in \M} (2 Z_i^* - 1) \epsilon_i + o_P(1)
\end{align*}
From this point, the same arguments used in the proof of \Cref{lemma:dm_critical_value} give $\sqrt{N_1} \hat{q}_{1 - \alpha}^{\reg} \xrightarrow{P} s z_{1 - \alpha/2}$.
\end{proof}

\subsection{Asymptotic distributions of test statistics}

In this section, we derive the asymptotic distributions of various test statistics.  Throughout, we assume that $g : \R^d \rightarrow \R^k$ is a bounded function and $\{ h_n \} \subset \R^k$ is a sequence of vectors satisfying $\| h_n \| \leq C n^{-1/2}$ for all $n$.  In addition, we fix a (possibly empty) set $S \subseteq [k]$ and define the following functions:
\begin{align}
\phi_S(x, z) &= \left\{ 
\begin{array}{ll}
(1, x, g_S(x)) &\text{if } S \neq \emptyset\\
(1, x) &\text{if } S = \emptyset.
\end{array}
\right. \label{eq:phi_S_definition}\\
\psi_S(x, z) &= \left\{ 
\begin{array}{ll}
(z - 0.5, 1, x, g_S(x)) &\text{if } S \neq \emptyset\\
(z - 0.5, 1, x) &\text{if } S = \emptyset.
\end{array}
\right. \label{eq:psi_S_definition}
\end{align}

\subsubsection{Unweighted regression statistic}

\begin{lemma}
\label{lemma:epsilon_difference_asymptotic_normality}
Assume $P \in H_0$ and set $\epsilon_i = Y_i - \mu(X_i)$.  Then $\sum_{Z_i = 1} (\epsilon_i - \epsilon_{m(i)}) / \sqrt{N_1} \rightsquigarrow N(0, s^2)$ where $s^2 = 2 \E \{ \sigma^2(X) \mid Z = 1 \}$.
\end{lemma}

\begin{proof}
Let $\F_n = \sigma(\{ (X_i, Z_i) \}_{i \leq n})$ and set $s_n^2 = \Var( \sum_{Z_i = 1} ( \epsilon_i - \epsilon_{m(i)}) / \sqrt{N_1} \mid \F_n) = \sum_{i \in \M} \sigma^2(X_i) / N_1$.  Then applying the Berry-Esseen theorem (\Cref{theorem:berry_esseen}) conditional on $\F_n$ gives:
\begin{align}
\begin{split}
\sup_{t \in \R} \left| \P \left( \frac{1}{s_n \sqrt{N_1}} \sum_{Z_i = 1} (\epsilon_i - \epsilon_{m(i)}) \leq t \, \bigg| \, \F_n \right) - \Phi(t) \right| &\leq \frac{C}{\sqrt{N_1}} \frac{1}{s_n^3} \frac{1}{N_1} \sum_{Z_i = 1} \E(| \epsilon_i - \epsilon_{m(i)}|^3 \mid \F_n)\\
&\leq \frac{4 C}{\sqrt{N_1}} \frac{1}{s_n^3} \frac{n}{N_1} \frac{1}{n} \sum_{i = 1}^n \E(| \epsilon_i|^3 \mid \F_n)\\
&\xrightarrow{P} 0. \label{eq:epsilon_difference_kolmogorov_convergence}
\end{split}
\end{align}
In particular, $\P( \sum_{Z_i = 1} ( \epsilon_i - \epsilon_{m(i)}) / s_n \sqrt{N_1} \leq t \mid \F_n) \rightarrow \Phi(t)$ for each fixed $t$.  Taking expectations on both sides with the bounded convergence theorem gives $\sum_{Z_i = 1} ( \epsilon_i - \epsilon_{m(i)}) / s_n \sqrt{N_1} \rightsquigarrow N(0, 1)$.  Since $s_n / s = 1 + o_P(1)$ by \Cref{lemma:sigma_star_limit}, Slutsky's theorem implies $\sum_{Z_i = 1} ( \epsilon_i - \epsilon_{m(i)}) / s \sqrt{N_1} \rightsquigarrow N(0, 1)$.  Multiplying both sides by $s$ gives the conclusion of the lemma.
\end{proof}

\begin{lemma} \label{lemma:paired_expansion}
Let $P \in H_0$ satisfy the assumptions of \Cref{theorem:paired_local_robustness}.  Define the vector $(\tilde \tau_n^{\reg}, \tilde \theta)$ as follows:
\begin{align}
( \tilde \tau_n^{\reg}, \tilde \theta) &= \argmin_{\tau, \theta} \sum_{i \in \M} \{ Y_i + h_n^{\top} g(X_i) - \tau(Z_i - 0.5) - \theta^{\top} \phi_S(X_i, Z_i) \}^2. \label{eq:paired_misspecified_regression}
\end{align}
Then $\sqrt{N_1} \tilde \tau_n^{\reg} = \sum_{i \in \M} (2 Z_i - 1) \epsilon_i / \sqrt{N_1} + o_P(1) \rightsquigarrow N(0, s^2)$.
\end{lemma}

\begin{proof}
Let $\mathbf{B} = \sum_{i \in \M} \psi_S(X_i, Z_i) \psi_S(X_i, Z_i)^{\top} / 2 N_1$ and $\mathbf{r} = \sum_{ \in \M} \{ \epsilon_i + h_{n,S^c}^{\top} g_{S^c}(X_i) \} \psi_S(X_i, Z_i) / 2 \sqrt{N_1}$.

By \Cref{lemma:B_invertibility} applied with $\phi(x) = ( x, g_S(x) )$, $\mathbf{B}$ is invertible with probability approaching one.  On this event, standard least-squares theory gives the formula $\sqrt{N_1} \tilde \tau_n^{\reg} = e_1^{\top} \mathbf{B}^{-1} \mathbf{r}$.  By \Cref{lemma:paired_design_matrix_properties} applied with $\phi(x) = (x, g_S(x))$, we have $e_1^{\top} \mathbf{B}^{-1} = (4, \mathbf{0}) + o_P(1)$.  By \Cref{lemma:ancova_covariance_properties}, $\mathbf{r} = O_P(1)$.  Putting these together gives:
\begin{align*}
\sqrt{N_1} \tilde \tau_n^{\reg} &= e_1^{\top} \mathbf{B}^{-1} \mathbf{r} + o_P(1)\\
&= (4, \mathbf{0})^{\top} \mathbf{r} + o_P(1) \mathbf{r} + o_P(1)\\
&= (4, \mathbf{0})^{\top} \frac{1}{2 \sqrt{N_1}} \sum_{i \in \M} \{ \epsilon_i + h_{n,S^c}^{\top} g_{S^c}(X_i) \} (Z_i - 0.5, 1, g_S(X_i)) + o_P(1) O_P(1)\\
&= \frac{1}{\sqrt{N_1}} \sum_{i \in \M} (2 Z_i - 1) \epsilon_i + \frac{1}{2 \sqrt{N_1}} \sum_{i \in \M} (2 Z_i - 1) h_{n,S^c}^{\top} g_{S^c}(X_i) + o_P(1)\\
&= \frac{1}{\sqrt{N_1}} \sum_{i \in \M} (2 Z_i - 1) \epsilon_i + \underbrace{\frac{1}{2} \sqrt{N_1} h_{n,S^c}^{\top}}_{= O_P(1)} \underbrace{\frac{1}{N_1} \sum_{Z_i = 1} \{ g_{S^c}(X_i) - g_{S^c}(X_{m(i)}) \}}_{= o_P(1) \text{ by \Cref{corollary:paired_discrepancy_Lp}}} + o_P(1)\\
&= \frac{1}{\sqrt{N_1}} \sum_{i \in \M} (2 Z_i - 1) \epsilon_i + o_P(1).
\end{align*}
The final expression in the preceding display converges in distribution to $N(0, s^2)$ by \Cref{lemma:epsilon_difference_asymptotic_normality} and Slutsky's theorem.
\end{proof}

\subsubsection{Weighted regression statistic}

\begin{lemma}
\label{lemma:replaced_design_matrix}
Let $\phi : \R^d \rightarrow \R^{\ell}$ satisfy $\Var \{ \phi(X) \mid Z = 1 \} \succ \mathbf{0}$ and $\E \{ \| \phi(X) \|^q \} < \infty$ for some $q > 4$, and define $\psi_{\phi}(x, z) = (z - 0.5, 1, \phi(x))$.  Set $W_i = Z_i + \sum_{j = 1}^n Z_j \mathbf{1} \{ m_r(j) = i \}$.  Then we have:
\begin{align*}
\mathbf{B}_r := \frac{1}{2 N_1} \sum_{i \in \M_r} W_i \psi_{\phi}(X_i, Z_i) \psi_{\phi}(X_i, Z_i)^{\top} \xrightarrow{P} 
\left[ 
\begin{array}{cc}
0.25 &\mathbf{0}^{\top}\\ 
\mathbf{0} &\E \{ (1, \phi(X))(1, \phi(X))^{\top} \mid Z = 1 \}
\end{array} 
\right] \succ \mathbf{0}.
\end{align*}
\end{lemma}

\begin{proof}
Apply \Cref{lemma:replaced_wlln} to each entry of $\mathbf{B}_r$.
\end{proof}

\begin{lemma}
\label{lemma:replaced_covariance_properties}
Let $\phi : \R^d \rightarrow \R^{\ell}$ satisfy $\Var \{ \phi(X) \mid Z = 1 \} \succ \mathbf{0}$ and $\E \{ \| \phi(X) \|^q \} < \infty$ for some $q > 4$, and define $\psi_{\phi}(x, z) = (z - 0.5, 1, \phi(x))$.  Define $\mathbf{r}_{r,n}$ as follows:
\begin{align*}
    \mathbf{r}_{r,n} &= \frac{1}{2 \sqrt{N_1}} \sum_{i \in \mathcal{M}_r} W_i \{ \epsilon_i + h_n^{\top} g(X_i) \} \psi_{\phi}(X_i, Z_i)
\end{align*}
where $W_i = Z_i + (1 - Z_i) \sum_{j = 1}^n Z_j \mathbf{1} \{ m_r(j) = i \}$.  Then $\mathbf{r}_{r,n} = O_P(1)$.
\end{lemma}

\begin{proof}
As in the proof of \Cref{lemma:ancova_covariance_properties}, we will actually prove a more general result:  for any function $f(x, z)$ with $\E \{ f(X, Z)^q \} < \infty$ for some $q > 4$, we will have $\sum_{i \in \mathcal{M}_r} W_i f(X_i, Z_i) \{ \epsilon_i + h_n^{\top} g(X_i) \} = O_P(1)$.  The conclusion of the lemma follows from this result by choosing $f$ to be the coordinates of the vector $\psi_{\phi}(x, z)$.

To prove the more general result, we start by considering the term $\sum_{i \in \mathcal{M}_r} W_i f(X_i, Z_i) \epsilon_i / 2 \sqrt{N_1}$:
\begin{align*}
\E \left( \left\{ \frac{1}{2 \sqrt{N_1}} \sum_{i \in \mathcal{M}_r} W_i f(X_i, Z_i) \epsilon_i \right\}^2 \, \bigg| \, \F_n \right) &= \frac{1}{4 N_1} \sum_{i \in \mathcal{M}_r} W_i^2 f(X_i, Z_i)^2 \sigma^2(X_i)\\
&\leq \frac{n}{4 N_1} \frac{1}{n} \sum_{i = 1}^n W_i^2 f(X_i, Z_i)^2 \sigma^2(X_i)
\end{align*}
Since $W_i \leq (1 + K_{i,n})$, $W_i^2$ has infinitely many moments that are uniformly bounded in $n$.  Moreover, $f(X, Z)^2$ has more than two moments, and $\sigma^2(X)$ has more than two moments.  Thus, their product $W_i^2 f(X_i, Z_i)^2 \sigma^2(X_i)$ has at least one bounded moment, and $\sum_{i = 1}^n W_i^2 f(X_i, Z_i)^2 \sigma^2(X_i) / n = O_P(1)$ by Markov's inequality.

Next, we consider the term $\sum_{i \in \mathcal{M}_r} W_i f(X_i, Z_i) h_n^{\top} g(X_i) / 2 \sqrt{N_1}$.  This term can be bounded as follows:
\begin{align*}
\left| \frac{1}{2 \sqrt{N_1}} \sum_{i \in \mathcal{M}_r} W_i f(X_i, Z_i) h_n^{\top} g(X_i) \right| &\leq \frac{\| h_n \| n}{2 \sqrt{N_1}} \frac{1}{n} \sum_{i = 1}^n W_i f(X_i, Z_i) \| g(X_i) \|\\
&\leq \frac{C}{2} \left( \frac{n}{N_1} \right)^{1/2} \frac{1}{n} \sum_{i = 1}^n W_i |f(X_i, Z_i)| \cdot \| g(X_i) \|.
\end{align*}
Since $\| g(X_i) \|$ has two moments and $W_i |f(X_i, Z_i)| \leq (1 + K_{i,n}) | f(X_i, Z_i)|$ has at least two moments, their product has at least one moment.  Thus, $\sum_{i = 1}^n W_i | f(X_i, Z_i)| \cdot \| g(X_i) \| /n = O_P(1)$ by Markov's inequality.

Combining the two terms proves our more general claim.
\end{proof}

\begin{lemma} \label{lemma:replaced_expansion}
Let $P$ satisfy the assumptions of \Cref{theorem:replaced_robustness}.  Consider the following weighted least-squares problem:
\begin{align}
(\tilde \tau_n^{\reg}, \tilde \theta) &= \argmin_{\tau, \theta} \sum_{i \in \M_r} W_i \{ Y_i + h_n^{\top} g(X_i) - \tau (Z_i - 0.5) - \theta^{\top} \phi_S(X_i) \}^2 \label{eq:replaced_misspecified_regression}
\end{align}
where $W_i = Z_i + (1 - Z_i) \sum_{j = 1}^n Z_j \mathbf{1} \{ m_r(j) = i \}$.  Then $\sqrt{N_1} \tilde \tau_n^{\reg} = \sum_{i \in \M_r} W_i (2Z_i - 1) \epsilon_i / \sqrt{N_1} + o_P(1)$.  Moreover, $\sqrt{N_1} \tilde \tau_n^{\reg} / \tilde s_n \rightsquigarrow N(0, 1)$ where $\tilde s_n^2 = \sum_{i \in \M_r} W_i^2 \sigma^2(X_i) / N_1$.
\end{lemma}

\begin{proof}
Let $\mathbf{B}_r = \sum_{i \in \M_r} W_i \psi_S(X_i, Z_i) \psi_S(X_i, Z_i)^{\top} / 2 N_1$ and $\mathbf{r}_{r,n} = \sum_{i \in \M_r} W_i \{ \epsilon_i + h_{n,S^c}^{\top} g_{S^c}(X_i) \} \psi_S(X_i, Z_i) / 2 \sqrt{N_1}$.

\Cref{lemma:replaced_design_matrix} (applied with $\phi(x) = (x, g_S(x))$) implies that $\mathbf{B}_r$ is invertible with probability tending to one, in which case standard least-squares formulae give the identity $\sqrt{N_1} \tilde \tau_n^{\reg} = e_1^{\top} \mathbf{B}_r^{-1} \mathbf{r}_{r,n}$.  Again applying \Cref{lemma:replaced_design_matrix} with $\phi(x) = (x, g_S(x))$ gives $e_1^{\top} \mathbf{B}^{-1} = (4, \mathbf{0}) + o_P(1)$.  \Cref{lemma:replaced_covariance_properties} gives $\mathbf{r}_{r,n} = O_P(1)$.  Thus, Slutsky's theorem gives:
\begin{align*}
\sqrt{N_1} \tilde \tau_n^{\reg} &= e_1^{\top} \mathbf{B}^{-1} \mathbf{r}_{r,n}\\
&= (4, \mathbf{0})^{\top} \mathbf{r}_{r,n} + o_P(1)\\
&= \frac{1}{\sqrt{N_1}} \sum_{i \in \M_r} W_i (2 Z_i - 1) \epsilon_i + \frac{1}{\sqrt{N_1}} \sum_{i \in \M_r} W_i(2 Z_i - 1) h_{n,S^c}^{\top} g_{S^c}(X_i) + o_P(1)\\
&= \frac{1}{\sqrt{N_1}} \sum_{i \in \M_r} W_i(2 Z_i - 1) \epsilon_i + \underbrace{\sqrt{N_1} h_{n,S^c}^{\top}}_{= O_P(1)} \underbrace{\frac{1}{N_1} \sum_{Z_i = 1} \{ g_{S^c(X_i)} - g_{S^c}(X_{m_r(i)}) \}}_{= o_P(1) \text{ by \Cref{corollary:empirical_Lp_convergence}}} + o_P(1)\\
&= \frac{1}{\sqrt{N_1}} \sum_{i \in \M_r} W_i(2Z_i - 1) \epsilon_i + o_P(1).
\end{align*}
Finally, we prove asymptotic normality.  Applying the Berry-Esseen theorem conditional on $\F_n = \sigma( \{ (X_i, Z_i) \}_{i \leq n})$ gives:
\begin{align*}
\sup_{t \in \R} \left| \P \left( \frac{1}{\tilde s_n \sqrt{N_1}} \sum_{i \in \M_r} W_i (2 Z_i - 1) \epsilon_i \leq t \, \bigg| \, \F_n \right) - \Phi(t) \right| &\leq \frac{C}{\sqrt{N_1}} \frac{1}{\tilde s_n^3} \frac{1}{N_1} \sum_{i \in \M_r} W_i^3 \E(| \epsilon_i|^3 \mid \F_n)\\
&\leq \frac{C}{\sqrt{N_1}} \frac{1}{\tilde s_n^3} \frac{1}{N_1} \sum_{i = 1}^n W_i^3 \E(| \epsilon_i|^3 \mid X_i, Z_i).
\end{align*}
Since $\tilde s_n^2 \geq \sum_{Z_i = 1} \sigma^2(X_i) / N_1 \xrightarrow{P} \E \{ \sigma^2(X) \mid Z = 1 \} > 0$, $1/ \tilde s_n^3 = O_P(1)$.  Since $W_i$ has infinitely many moments (\Cref{lemma:Kn_moments}) and $\E(| \epsilon_i|^3 \mid X_i, Z_i)$ has more than one moment (\Cref{assumption:primitives}.\ref{item:four_moments}), their product has at least one moment so $\sum_{i = 1}^n W_i^3 \E(| \epsilon_i|^3 \mid X_i, Z_i) / n = O_P(1)$ by Markov's inequality.  Thus, the upper bound in the preceding display tends to zero in probability.  Taking expectations on both sides with the bounded convergence theorem gives $\sum_{i \in \M_r} W_i (2 Z_i - 1) \epsilon_i / \tilde s_n \sqrt{N_1} \rightsquigarrow N(0, 1)$.  By Slutsky's theorem, this proves $\sqrt{N_1} \tilde \tau_n^{\reg} / \tilde s_n \rightsquigarrow N(0, 1)$ as well.
\end{proof}

\subsection{Asymptotics of robust standard errors} 

\subsubsection{Unweighted regressions} \label{section:paired_robust_standard_errors}

Throughout this section, we assume the following:  $P \in H_0$ satisfies the assumptions of \Cref{theorem:paired_local_robustness}, $g : \R^d \rightarrow \R^k$ is a bounded function, and $\{ h_n \} \subset \R^k$ is a sequence with $\| h_n \| \leq C n^{-1/2}$.  We also fix a (possibly empty) set $S \subseteq [k]$ and use the notation $\phi_S, \psi_S$ from Equations (\ref{eq:phi_S_definition}) and (\ref{eq:psi_S_definition}).  We also define pseudo-residuals $\tilde \epsilon_i$ as follows:
\begin{align}
\tilde \epsilon_i = Y_i + h_n^{\top} g(X_i) - \tilde \tau_n^{\reg} Z_i - \tilde \theta^{\top} \phi_S(X_i) \label{eq:epsilon_tilde_definition}
\end{align}
where $(\tilde \tau_n^{\reg}, \tilde \theta)$ are defined in \Cref{eq:paired_misspecified_regression}.  For all results apart from \Cref{theorem:model_dependence}, we will only need these quantities when $S = \emptyset$.

\begin{lemma} \label{lemma:paired_regression_consistency}
Let $(\tilde \tau_n^{\reg}, \tilde \theta)$ be as in \Cref{eq:paired_misspecified_regression}. 
\begin{itemize}[itemsep=-1ex]
    \item If $S \neq \emptyset$ and $\Var \{ (X, g(X)) \mid Z = 1 \} \succ \mathbf{0}$, then $\tilde \tau_n^{\reg} = o_P(1)$ and $\| \tilde \theta - (\gamma, \beta, \mathbf{0}) \| = o_P(1)$.
    \item If $S = \emptyset$, then $\tilde \tau_n^{\reg} = o_P(1)$ and $\| \tilde \theta - (\gamma, \beta) \| = o_P(1)$.
\end{itemize}
\end{lemma}

\begin{proof}
Let $\mathbf{B} = \sum_{i \in \M} \psi_S(X_i, Z_i) \psi_S(X_i, Z_i)^{\top} / 2 N_1$ and $\mathbf{r} = \sum_{i \in \M} \{ \epsilon_i + h_{n,S^c}^{\top} g_{S^c}(X_i) \} \psi_S(X_i)$.

Suppose first that $S \neq \emptyset$.  Then by \Cref{lemma:paired_design_matrix_properties} applied with $\phi(x) = (x, g_S(x))$, the matrix $\mathbf{B}$ is invertible with probability tending to one.  On this event, standard least-squares theory gives the identity $(\tilde \tau_n^{\reg}, \tilde \theta) - (0, \gamma, \beta, h_{n,S}) = \mathbf{B}^{-1} \mathbf{r} / \sqrt{N_1}$.  Since $\mathbf{B}^{-1} = O_P(1)$ by \Cref{lemma:paired_design_matrix_properties} and $\mathbf{r} = O_P(1)$ by \Cref{lemma:ancova_covariance_properties}, we may conclude that $\| ( \tilde \tau_n^{\reg}, \tilde \theta) - (0, \gamma, \beta, h_{n,S}) \| = O_P(1/\sqrt{N_1}) = o_P(1)$.  Since $h_{n,S} \rightarrow 0$, this also implies $\| (\tilde \tau_n^{\reg}, \tilde \theta) - (0, \gamma, \beta, \mathbf{0}) \| = o_P(1)$.

With suitable changes of notation, the same proof goes through when $S = \emptyset$.
\end{proof}

\begin{lemma} \label{lemma:paired_sandwich_filling}
Let $\mathbf{A} = \sum_{i \in \M} \tilde \epsilon_i^2 \psi_S(X_i, Z_i) \psi_S(X_i, Z_i)^{\top}$.  Then $\mathbf{A}_{11} = 0.25 \E \{ \sigma^2(X) \mid Z = 1 \} + o_P(1)$ and $\mathbf{A} = O_P(1)$.
\end{lemma}

\begin{proof}
Let $\tilde \vartheta = (\tilde \tau_n^{\reg}, \tilde \theta)$.  Let $\vartheta = (0, \gamma, \beta, \mathbf{0})$ if $S \neq \emptyset$ and otherwise set $\vartheta = (0, \gamma, \beta)$.  Then we may write the following:
\begin{align*}
\mathbf{A}_{11} &= \frac{1}{8 N_1} \sum_{i \in \M} \tilde \epsilon_i^2\\
&= \frac{1}{8 N_1} \sum_{i \in \M} \{ Y_i + h_n^{\top} g(X_i) - \tilde \vartheta^{\top} \psi_S(X_i, Z_i) \}^2\\
&= \underbrace{\frac{1}{8 N_1} \sum_{i \in \M} \{ Y_i + h_n^{\top} g(X_i) - \vartheta^{\top} \psi_S(X_i, Z_i) \}^2}_a + \underbrace{\frac{1}{8 N_1} \sum_{i \in \M} \{ 2 Y_i + 2 h_n^{\top} g(X_i) - (\vartheta + \tilde \vartheta)^{\top} \psi_S(X_i) \} (\vartheta - \tilde \vartheta)^{\top} \psi_S(X_i, Z_i)}_b
\end{align*}
Term $a$ converges to $0.25 \E \{ \sigma^2(X) \mid Z = 1 \}$ by the following calculation:
\begin{align*}
a &= \frac{1}{8 N_1} \sum_{i \in \mathcal{M}} \{ \epsilon_i + h_n^{\top} g(X_i) \}^2\\
&= \frac{1}{8 N_1} \sum_{i \in \M} \epsilon_i^2 + \underbrace{\frac{1}{8 N_1} \sum_{i \in \M} 2 \epsilon_i h_n^{\top} g(X_i)}_{= o_P(1) \text{ by \Cref{lemma:conditional_wlln} conditional on } \{ (X_i, Z_i) \}_{i \leq n}} + \underbrace{\frac{1}{8 N_1} \sum_{i \in \M} \{ h_n^{\top} g(X_i) \}^2}_{\leq \| h_n \|^2 \sum_{i = 1}^n \| g(X_i) \|^2 / 2 N_1 = O_P(1/n)}\\
&= \frac{1}{8 N_1} \sum_{i \in \M} \epsilon_i^2 + o_P(1)\\
&= \frac{1}{8 N_1} \sum_{i \in \M} \sigma^2(X_i) + \underbrace{\frac{1}{8 N_1} \sum_{i \in \M} \{ \epsilon_i^2 - \sigma^2(X_i) \}}_{= o_P(1) \text{ by \Cref{lemma:conditional_wlln} conditional on $\{ (X_i, Z_i) \}_{i \leq n}$}} + o_P(1)\\
&= \frac{1}{8 N_1} \sum_{Z_i = 1} 2 \sigma^2(X_i) + \underbrace{\frac{1}{8 N_1} \sum_{Z_i = 1} \{ \sigma^2(X_{m(i)}) - \sigma^2(X_i) \}}_{= o_P(1) \text{ by \Cref{corollary:empirical_Lp_convergence}}} + o_P(1)\\
&\xrightarrow{P} 0.25 \E \{ \sigma^2(X) \mid Z = 1 \}.
\end{align*}
Meanwhile, term $b$ tends to zero by the Cauchy-Schwarz inequality and the fact that $\| \vartheta - \tilde \vartheta \| \xrightarrow{P} 0$ (\Cref{lemma:paired_regression_consistency}).  Thus, $\mathbf{A}_{11} = 0.25 \E \{ \sigma^2(X) \mid Z = 1 \} + o_P(1)$.

The proof that $\mathbf{A} = O_P(1)$ follows by applying similar arguments to each of the entries of $\mathbf{A}$.  The details are omitted.
\end{proof}

\begin{lemma} \label{lemma:paired_hc_standard_error}
Define the following quantity:
\begin{align*}
\tilde \sigma^2_{\hc} &= \left( \sum_{i \in \M} \psi_S(X_i, Z_i) \psi_S(X_i, Z_i)^{\top} \right)^{-1} \left( \sum_{i \in \M} \tilde \epsilon_i^2 \psi_S(X_i, Z_i) \psi_S(X_i, Z_i)^{\top} \right) \left( \sum_{i \in \M} \psi_S(X_i, Z_i) \psi_S(X_i, Z_i)^{\top} \right)^{-1}.
\end{align*}
Then $N_1 \tilde \sigma_{\hc}^2 \xrightarrow{P} 2 \E \{ \sigma^2(X) \mid Z = 1 \}$.
\end{lemma}

\begin{proof}
Let $\mathbf{B} = \sum_{i \in \M} \psi_S(X_i, Z_i) \psi_S(X_i, Z_i)^{\top} / 2 N_1$ and $\mathbf{A} = \sum_{i \in \M} \tilde \epsilon_i^2 \psi_S(X_i, Z_i) \psi_S(X_i, Z_i)^{\top} / 2 N_1$.  Then we have the identity $N_1 \tilde \sigma^2_{\hc} = 0.5 e_1^{\top} \mathbf{B}^{-1} \mathbf{A} \mathbf{B}^{-1} e_1$.  \Cref{lemma:paired_design_matrix_properties} shows that $\mathbf{B}^{-1} e_1 = (4, \mathbf{0}) + o_P(1)$. \Cref{lemma:paired_sandwich_filling} shows that $\mathbf{A} = O_P(1)$.  Thus, $N_1 \tilde \sigma^2_{\hc} = 8 \mathbf{A}_{11} + o_P(1)$.  Finally, \Cref{lemma:paired_sandwich_filling} shows that $8 \mathbf{A}_{11} = 2 \E \{ \sigma^2(X) \mid Z = 1 \} + o_P(1)$.
\end{proof}

\subsubsection{Weighted regressions}

Throughout this section, we make similar assumptions as in \Cref{section:paired_robust_standard_errors}.  We assume $P \in H_0$ satisfies the assumptions of \Cref{theorem:replaced_robustness}, $g : \R^d \rightarrow \R^k$ is a bounded function, and $\{ h_n \} \subset \R^k$ is a sequence with $\| h_n \| \leq C n^{-1/2}$.  We also fix a (possibly empty) set $S \subseteq [k]$ and define $\phi_S, \psi_S$ as in Equations (\ref{eq:phi_S_definition}), (\ref{eq:psi_S_definition}) and (\ref{eq:epsilon_tilde_definition}), respectively.  Finally, we define $\tilde \epsilon_i$ as in \Cref{eq:epsilon_tilde_definition}, except the regression coefficients $(\tilde \tau_n^{\reg}, \tilde \theta)$ in that display now come from the weighted regression \Cref{eq:replaced_misspecified_regression} instead of the ordinary regression from \Cref{eq:paired_misspecified_regression}.

\begin{lemma} \label{lemma:replaced_regression_consistency}
Let $(\tilde \tau_n^{\reg}, \tilde \theta)$ be as in \Cref{eq:replaced_misspecified_regression}.
\begin{itemize}[itemsep=-1ex]
    \item If $S \neq \emptyset$ and $\Var \{ (X, g(X)) \mid Z = 1 \} \succ \mathbf{0}$, then $\tilde \tau_n^{\reg} = o_P(1)$ and $\| \tilde \theta - (\gamma, \beta, \mathbf{0}) \| = o_P(1)$.
    \item If $S = \emptyset$, then $\tilde \tau_n^{\reg} = o_P(1)$ and $\| \tilde \theta - (\gamma, \beta) \| = o_P(1)$.
\end{itemize}
\end{lemma}

\begin{proof}
The proof is identical to that of \Cref{lemma:paired_regression_consistency}, except \Cref{lemma:replaced_design_matrix} and \Cref{lemma:replaced_covariance_properties} are used  in place of \Cref{lemma:paired_design_matrix_properties} and \Cref{lemma:ancova_covariance_properties}, respectively.
\end{proof}

\begin{lemma} \label{lemma:replaced_sandwich_filling}
Let $\mathbf{A}_r = \sum_{i \in \M_r} W_i^2 \tilde \epsilon_i^2 \psi_S(X_i, Z_i) \psi_S(X_i, Z_i)^{\top} / 2 N_1$.  Then $\mathbf{A}_{r,11} = \sum_{i \in \M_r} W_i^2 \sigma^2(X_i) / 8 N_1 + o_P(1)$ and $\mathbf{A}_r = O_P(1)$.
\end{lemma}

\begin{proof}
The proof is similar to that of \Cref{lemma:paired_sandwich_filling}.  Let $\tilde \vartheta = ( \tilde \tau_n^{\reg}, \tilde \theta)$.  Let $\vartheta = (0, \gamma, \beta, \mathbf{0})$ if $S \neq \emptyset$ and otherwise set $\vartheta = (0, \gamma, \beta)$.  Then the following holds:
\begin{align*}
\mathbf{A}_{r,11} &= \frac{1}{8 N_1} \sum_{i \in \M_r} W_i^2 \tilde \epsilon_i^2\\
&= \frac{1}{8 N_1} \sum_{i \in \M_r} W_i^2 \{ Y_i + h_n^{\top} g(X_i) - \tilde \vartheta^{\top} \psi_S(X_i, Z_i) \}^2\\
&= \underbrace{\frac{1}{8 N_1} \sum_{i \in \M_r} W_i^2 \{ Y_i + h_n^{\top} g(X_i) - \vartheta^{\top} \psi_S(X_i, Z_i) \}^2}_a\\
&+ \underbrace{\frac{1}{8 N_1} \sum_{i \in \M_r} W_i^2 \{ 2 Y_i + 2 h_n^{\top} g(X_i) - (\vartheta + \tilde \vartheta)^{\top} \psi_S(X_i, Z_i) \} (\vartheta - \tilde \vartheta)^{\top} \psi_S(X_i, Z_i)}_b
\end{align*}
The calculation below shows that term $a$ has the desired behavior.  We freely use the fact that $W_i$ has infinitely many moments, which is assured by \Cref{lemma:Kn_moments}.
\begin{align*}
a &= \frac{1}{8 N_1} \sum_{i \in \M_r} W_i^2 \{ \epsilon_i + h_n^{\top} g(X_i) \}^2\\
&= \frac{1}{8 N_1} \sum_{i \in \M_r} W_i^2 \epsilon_i^2 + \underbrace{\frac{1}{8 N_1} \sum_{i \in \M_r} 2 W_i^2 \epsilon_i h_n^{\top} g(X_i)}_{= o_P(1) \text{ by \Cref{lemma:conditional_wlln} conditional on $\{ (X_i, Z_i) \}_{i \leq n}$}} + \underbrace{\frac{1}{8 N_1} \sum_{i \in \M_r} W_i^2 \{ h_n^{\top} g(X_i) \}^2}_{\leq \| h_n \|^2 \sum_{i \in \M_r} W_i^2 \| g(X_i) \|^2 / 8 N_1 = O_P(1/n)}\\
&= \frac{1}{8 N_1} \sum_{i \in \M_r} W_i^2 \sigma^2(X_i) + \underbrace{\frac{1}{8 N_1} \sum_{i \in \M_r} W_i^2 \{ \epsilon_i^2 - \sigma^2(X_i) \}}_{= o_P(1) \text{ by \Cref{lemma:conditional_wlln} conditional on $\{ (X_i, Z_i) \}_{i \leq n}$}} + o_P(1)\\
&= \frac{1}{8 N_1} \sum_{i \in \M_r} W_i^2 \sigma^2(X_i) + o_P(1).
\end{align*}
The proof that $\mathbf{A}_r = O_P(1)$ follows by applying similar arguments to each entry of $\mathbf{A}_r$.  As in the proof of \Cref{lemma:paired_sandwich_filling}, we omit the details.
\end{proof}

\begin{lemma} \label{lemma:replaced_hc_standard_error}
Define $\tilde \sigma^2_{r,\hc}$ as follows:
\begin{align*}
\tilde \sigma^2_{r, \hc} &= \left( \sum_{i \in \M_r} W_i \psi_S(X_i, Z_i) \psi_S(X_i, Z_i)^{\top} \right)^{-1} \left( \sum_{i \in \M_r} W_i^2 \tilde \epsilon_i^2 \psi_S(X_i, Z_i) \psi_S(X_i, Z_i)^{\top} \right) \left( \sum_{i \in \M_r} W_i \psi_S(X_i, Z_i) \psi_S(X_i, Z_i)^{\top} \right)^{-1}.
\end{align*}
Then $N_1 \tilde \sigma^2_{r, \hc} = \sum_{i \in \M_r} W_i \sigma^2(X_i) / N_1 + o_P(1)$.
\end{lemma}

\begin{proof}
Let $\mathbf{B}_r = \sum_{i \in \M_r} W_i \psi_S(X_i, Z_i) \psi_S(X_i, Z_i)^{\top} / 2 N_1$ and $\mathbf{A}_r = \sum_{i \in \M_r} W_i^2 \tilde \epsilon_i^2 \psi_S(X_i, Z_i) \psi_S(X_i, Z_i)^{\top} / 2 N_1$.  Then $N_1 \tilde \sigma^2_{r, \hc} = 0.5 e_1^{\top} \mathbf{B}_r^{-1} \mathbf{A}_r \mathbf{B}_r$.  By \Cref{lemma:replaced_covariance_properties}, $\mathbf{B}_r e_1 = (4, \mathbf{0}) + o_P(1)$ and by \Cref{lemma:replaced_sandwich_filling}, $\mathbf{A}_r = O_P(1)$.  Thus, $N_1 \tilde \sigma_{r, \hc}^2 = 8 \mathbf{A}_{r,11} + o_P(1)$.  Finally, \Cref{lemma:replaced_sandwich_filling} shows that $8 \mathbf{A}_{r,11} = \sum_{i \in \M_r} W_i^2 \sigma^2(X_i) / N_1 + o_P(1)$.
\end{proof}

\section{Additional proofs} \label{section:additional_proofs}

In this section, we prove \Cref{lemma:Kn_moments}, which shows that $K_{i,n}$ (the number of times observation $i$ is used as an untreated match under the matching-with-replacement scheme (\ref{eq:matching_with_replacement})) has uniformly bounded moments of all orders.  Throughout, we assume that the distribution $P$ satisfies \Cref{assumption:primitives}.

\subsection{Notation}

In this section, we use the following notations.  We let $\mathcal{F}_n = \{ (X_i, Z_i, U_i) \}_{i \leq n}$ denote the information used in forming matches, where $U_i \sim \textup{Uniform}(0, 1)$ is independent randomness used in the tie-breaking scheme of \cite{devroye_etal1994} which works as follows:  among untreated observations $j$ with the minimal value of $d_M(X_i, X_j)$, we choose the one which minimizes $|U_i - U_j|$ as a match.  For matches, we use the following notation, which is slightly more descriptive than what was used in the main text:
\begin{align*}
m_r(i, \F_n) &= \argmin_{Z_j = 0} d_M(X_i, X_j)\\
L_{i,n}(\F_n) &= \sum_{j = 1}^n \mathbf{1} \{ m_r(j, \F_n) = i \}\\
K_{i,n}(\F_n) &= \sum_{j = 1}^n Z_j \mathbf{1} \{ m_r(j, \F_n) = i \}
\end{align*}
(here, we recall $d_M$ is the estimated Mahalanobis distance defined in \Cref{lemma:mahalanobis_euclidean_equivalence}).  In words, $m_r(i, \F_n)$ is the nearest untreated neighbor of observation $i$.  When there are no untreated observations in the dataset $\F_n$, we arbitrarily set $m_r(i, \F_n) = 0$.  Meanwhile, $L_{i,n}(\F_n)$ counts the number of times observation $i$ is the nearest untreated neighbor of another observation, and $K_{i,n}(\F_n)$ counts the number of times observation $i$ is used as an untreated match.  We also use the abbreviation ``NN" for ``nearest neighbors."

\subsection{Preparation}

\begin{lemma}
\label{lemma:maximum_neighbors}
Let $\gamma_d < \infty$ be the minimal number of cones centered at the origin of angle $\pi/6$ that cover $\R^d$.  Then for any distinct points $x_1, \ldots, x_n \in \R^d$, we have:
\begin{align*}
\sum_{j = 1}^n \mathbf{1} \{ \text{$x_1$ is among the $k$-NN of $x_j$ in $\{ x_i \}_{i \in [n] \backslash \{ j \}}$} \} \leq k \gamma_d.
\end{align*}
\end{lemma}

\begin{proof}
This follows from \cite[Corollary 11.1]{devroye_etal1996} after replacing $x_i$ by $\hat{\mathbf{\Sigma}}^{-1/2} x_i$, where we recall the convention that $\hat{\mathbf{\Sigma}}^{-1} = \mathbf{I}_{d \times d}$ when $\hat{\mathbf{\Sigma}}$ is singular.
\end{proof}

\begin{lemma}
\label{lemma:coupling}
Independently of $\F_n$, let $Z_2', \ldots, Z_n' \overset{\text{iid}}{\sim} \textup{Bernoulli}(1 - \delta)$.  Let $\F_n'$ be the dataset $\F_n$ except with $Z_i$ replaced by $Z_i'$ for all $i \geq 2$.  Then $\E \{ L_{1,n}(\F_n')^q \} \leq \E \{ L_{1,n}(\F_n')^q \}$ for all $q > 0$.
\end{lemma}

\begin{proof}
We prove this by coupling the distributions of $\F_n$ and $\F_n'$.  On some probability space, construct $n$ independent random vectors $(X_i, U_i, V_i) \sim P_X \times \textup{Uniform}(0, 1) \times \textup{Uniform}(0, 1)$.  Define $Z_i = \mathbf{1} \{ V_i > P(Z_i = 0 \mid X_i) \}$ and $Z_i' = \mathbf{1} \{ V_i > \delta \}$.  Since $P(Z = 0 \mid X) \geq \delta$ almost surely (\Cref{assumption:primitives}.\ref{item:overlap}), $Z_i' \geq Z_i$ with probability one.  With this coupling, it is easy to see that $L_{1,n}(\F_n) \leq L_{1,n}(\F_n')$ almost surely.  This is because changing some ``untreated" units to ``treated" only decreases the number of ``competitors" of observation 1.  Hence, $L_{1,n}(\F_n) \leq L_{1,n}(\F_n')$ and the conclusion follows.
\end{proof}

\begin{lemma}
\label{lemma:multiple_matches_probability_estimate}
Let $\F_n'$ be as in \Cref{lemma:coupling}.  Then for all $k \geq 2$, we have:
\begin{align}
\P\{ m_r(2, \F_n') = \ldots = m_r(k, \F_n') = 1 \mid Z_1 = 0, N_0 \} \leq \left( \frac{k \gamma_{d + 1}}{N_0} \right)^{k - 1} \label{eq:multiple_matches_probability_estimate}
\end{align}
\end{lemma}

\begin{proof}
Throughout this proof, we simply write $m_r(i)$ in place of $m_r(i, \F_n')$.

Assume first that $X_1, \ldots, X_n$ are almost surely distinct.  

In that case, $m_r(i) = 1$ can only happen if $Z_i = 1$;  otherwise, $X_i$ would be its own nearest untreated neighbor.  Therefore, whenever $(k - 1) + N_0 > n$, one of the $X_i$'s must have $Z_i = 0$ by the pidgeonhole principle, and hence $\P\{ m_r(2) = \ldots = m_r(k) = 1 \mid Z_1 = 0, N_0 \} = 0$.  Clearly (\ref{eq:multiple_matches_probability_estimate}) holds in this case.

Now we consider the more interesting case where $(k - 1) + N_0 \leq n$.  In this case, we will show that the bound holds even conditionally on $(X_1, U_1)$.  Begin by conditioning on the event $Z_2 = \ldots = Z_k = 1$ and translating the event into the language of nearest neighbors:
\begin{align*}
&\P \{ m_r(2) = \ldots = m_r(k) = 1 \mid X_1, U_1, Z_1 = 0, N_0) \}\\
&\leq \P\{ m_r(2) = \ldots = m_r(k) = 1 \mid X_1, U_1, Z_1 = 0, N_0, Z_2 = \ldots = Z_k = 1 \}\\
&= \P(  \text{$X_1$ is the NN of $X_j$ in $\{ X_i \}_{Z_i = 0}$ for all $j \in \{ 2, \ldots, k \}$} \mid X_1, U_1, Z_1 = 0, N_0, Z_2 = \ldots = Z_k = 1 )
\end{align*}
By symmetry, the probability in the upper bound is the same no matter which $N_0 - 1$ observations in $\{ k + 1, \ldots, n \}$ are the ones with $Z_i = 0$.  Therefore, we may as well assume they are $k + 1, \ldots, k + N_0 - 1$, which gives the bound:
\begin{align*}
&\P \{ m_r(2) = \ldots = m_r(k) = 1 \mid X_1, U_1, Z_1 = 0, N_0 \}\\
&\leq \P [ \forall j \in [k] \backslash \{ 1 \}\text{, $X_1$ is the NN of $X_j$ in $\{ X_i \}_{i \in [k-1+N_0] \backslash \{ 2, \ldots, k \}}$} \mid X_1, U_1, Z_1 = 0, Z_2 = \ldots = Z_k = 1, N_0]\\
&= \P[ \forall j \in [k] \backslash \{ 1 \}\text{, $X_1$ is the NN of $X_j$ in $\{ X_i \}_{i \in [k - 1 + N_0] \backslash \{ 2, \ldots, k \}}$} \mid X_1, U_1, Z_1 = 0, N_0]\\
&\leq \P[ \forall j \in [k] \backslash \{ 1 \} \text{, $X_1$ is among the $k$-NN of $X_r$ in $\{ X_i \}_{i \in [k - 1 + N_0] \backslash \{ j \}}$} \mid X_1, U_1, Z_1 = 0, N_0].
\end{align*}
Again by symmetry, the probability in the final line of the preceding display would be the same if we replaced $[k] \backslash \{ 1 \}$ with any other set of $k - 1$ distinct indices in $[k - 1 + N_0 ] \backslash \{ 1 \}$.  Combining this observation with \Cref{lemma:maximum_neighbors} gives the desired bound with $\gamma_d$ in place of $\gamma_{d + 1}$:
\begin{align*}
&\P[ m_r(2) = \ldots = m_r(k) = 1 \mid X_1, U_1, Z_1 = 0, N_0]\\
&= \sum_{1 < i_1 < \ldots <i_{k-1} \leq k - 1 + N_0} \frac{\P[ X_1 \text{ is among the $k$-NN of $X_{i_r}$ in $\{ X_i \}_{i \in [k - 1 + N_0] \backslash \{ i_r \}}$ for each $i_r$} \mid X_1, U_1, Z_1 = 0, N_0]}{(N_0 + k - 1) \times (N_0 + k - 2) \times \cdots \times (N_0)}\\
&= \E \left(  \frac{\sum_{1 < i_1 < \ldots < i_{k-1} \leq k - 1 + N_0} \mathbf{1} \{ \text{$X_1$ is among the $k$-NN of $X_{i_r}$ in $\{ X_i \}_{i \in [k - 1 + N_0] \backslash \{ i_r \}}$ for each $i_r$ \}}}{(N_0 + k - 1) \times (N_0 + k - 2) \times \cdots \times (N_0)} \, \bigg| \, X_1, U_1, Z_1 = 0, N_0 \right)\\
&\leq \E \left\{ \frac{\left( \sum_{1 < i \leq k -1 + N_0} \mathbf{1} \{ \text{$X_1$ is among the $k$-NN of $X_i$ in $\{ X_i \}_{i \in [k - 1 + N_0] \backslash \{ i \}}$} \} \right)^{k -1}}{(N_0 + k - 1) \cdots (N_0)} \, \bigg| \, X_1, U_1, Z_1 = 0, N_0 \right\}\\
&\leq \frac{(k \gamma_d)^{k - 1}}{(N_0 + k - 1) \cdots (N_0)}\\
&\leq \left( \frac{k \gamma_d}{N_0} \right)^{k - 1}.
\end{align*}
This proves the result in the case where the $X_i$'s are almost surely distinct.

We only sketch the extension to the general case, which closely follows \cite{devroye_etal1994}.  Replace $X_i$ by $\{ X_i, r(\epsilon) U_i \}$ where $r(\epsilon)$ is so small that with probability at least $1 - \epsilon$, nearest-neighbor matchinb based on the $\tilde X_i$'s is the same as matching based on the $X_i$'s and then tie-breaking using the $U_i$'s (this can always be achieved by choosing $r(\epsilon)$ small enough so that all nonzero differences $\| X_i - X_j \|$ are larger than $r(\epsilon)$ with probability at least $1 - \epsilon$).  Then, but for an additive slippage of $\epsilon$, the above result applies to the $\tilde X_i$'s but with $\gamma_{d + 1}$ instead of $\gamma_d$ since $\tilde X_i \in \R^{d + 1}$.  Finally, take $\epsilon$ down to zero.
\end{proof}

\subsection{Proof of \Cref{lemma:Kn_moments}}

\begin{proof}
Since $K_{1,n}(\F_n)^q \leq L_{1,n}(\F_n)^q$, it suffices to prove a bound for $L_{1,n}(\F_n)$.  Moreover, Jensen's inequality implies we only need to consider integer values of $q$.

By \Cref{lemma:coupling}, $\E \{ L_{1,n}(\F_n)^q \} \leq \E \{ L_{1,n}(\F_n')^q \mid Z_1 = 0 \}$.  We control the expectation in this upper bound by first conditioning on $N_0$:
\begin{align*}
\E \{ L_{1,n}(\F_n') \mid Z_1 = 0, N_0 \} &= \E \left[ \left( 1 + \sum_{i = 2}^n \mathbf{1} \{ m_r(2, \F_n') = 1 \} \right)^q \, \bigg| \, Z_1 = 0, N_0 \right]\\
&\leq 2^{q - 1} \left( 1 + \E \left[ \left( \sum_{i = 2}^n \mathbf{1} \{ m_r(2, \F_n') = 1 \} \right)^q \, \bigg| \, Z_1 = 0, N_0 \right] \right)\\
&\leq 2^{q - 1} \left\{ 1 + \sum_{2 \leq i_1, \ldots, i_q \leq n} \P[ m_r(i_1, \F_n') = \ldots = m_r(i_q, \F_n') = 1 \mid Z_1 = 0, N_0] \right\}
\end{align*}
By exchangeability of the observations, $\P \{ m_r(i_1, \F_n') = \ldots = m_r(i_q, \F_n') = 1 \mid Z_1 = 0, N_0 \}$ depends only on the number of distinct indices $(i_1, \ldots, i_q)$ and not on the identity of those indices.  For any $\ell \leq q$, the number of sequences $(i_1, \ldots, i_q) \in [n - 1]^q$ with $\ell$ distinct indices is at most $\ell^q \binom{n - 1}{\ell}$.  Thus, \Cref{lemma:multiple_matches_probability_estimate} gives the further bound:
\begin{align*}
\E \{ L_{1,n}( \F_n') \mid Z_1 = 0, N_0 \} &\leq 2^{q - 1} \left\{ 1 + \sum_{\ell = 1}^q \ell^q \binom{n - 1}{\ell} \P [ m_r(2, \F_n') = \ldots = m_r(\ell + 1, \F_n') = 1 \mid Z_1 = 0, N_0 ] \right\}\\
&\leq 2^{q - 1} \left[ 1 + \sum_{\ell = 1}^q \ell^q \binom{n - 1}{\ell} \left\{ \frac{(\ell + 1) \gamma_{d + 1}}{N_0} \right\}^{\ell}  \right]\\
&\leq 2^{q - 1} \left[ 1 + \sum_{\ell = 1}^q \ell^q \{ \gamma_{d + 1}(\ell + 1) \}^{\ell}(n - 1)^{\ell} N_0^{-\ell} \right].
\end{align*}
Now, we take expectations over $N_0$ on both sides of the preceding display.  The conditional distribution of $N_0$ given $Z_1 = 0$ stochastically dominates that of $1 + N_0$ where $N_0' \sim \textup{Bernoulli}(n - 1, \delta)$.  By standard binomial concentration, it can be shown that for any $\ell$, there exists a constant $C(\ell) < \infty$ and $N_{\ell} \geq 1$ such that $\E\{ (1 + N_0')^{-\ell} \} \leq C(\ell) / \{ \delta (n - 1) \}^{\ell}$ for all $n \geq N_{\ell}$.  Therefore, we may conclude:
\begin{align*}
\E \{ L_{1,n}( \F_n')^q \mid Z_1 = 0 \} &\leq 2^{q - 1} \left[ 1 + \sum_{\ell = 1}^q \ell^q \{ \gamma_{d + 1} (\ell + 1) \}^{\ell} (n - 1)^{\ell} \E \{ (1 + N_0')^{- \ell} \} \right]\\
&\leq 2^{q - 1} \left[ 1 + \sum_{\ell = 1}^q \ell^q \{ \gamma_{d + 1}(\ell + 1) \}^{\ell} (n - 1)^{\ell} \frac{C(\ell)}{\delta^{\ell} (n - 1)^{\ell}} \right]\\
&\leq 2^{q - 1} \left[ 1 + \sum_{\ell = 1}^q C(\ell) \ell^q \{ \gamma_{d + 1}(\ell + 1)/\delta) \}^{\ell} \right].
\end{align*}
Since this upper bound does not depend on $n$ and holds for all large $n$, we conclude that $\E \{ L_{1,n}(\F_n')^q \mid Z_1 = 0 \}$ is uniformly bounded in $n$.  Since $\E \{ K_{1,n}(\F_n)^q \} \leq \E \{ L_{1,n}(\F_n)^q \} \leq \E \{ L_{1,n}(\F_n')^q \} \leq \E \{ L_{1,n}(\F_n')^q \mid Z_1 = 0 \}$, this proves the result.
\end{proof}

\end{document}